\documentclass[a4paper,twoside,12pt]{article}
\usepackage{pict2e}
\usepackage{curve2e}
\usepackage{graphicx}
\usepackage{color}
\usepackage{dsfont}
\usepackage{ifthen}
\usepackage{enumitem}
\usepackage{amssymb}

\makeatletter
\@addtoreset{equation}{section}
\@addtoreset{figure}{section}
\makeatother

\newtheorem{prop}{\bf Proposition}[section]
\newtheorem{lem}[prop]{\bf Lemma}
\newtheorem{cor}[prop]{\bf Corollary}

\newtheorem{thm}[prop]{\bf Theorem}
\newtheorem{rem}[prop]{\bf Remark}
\newtheorem{cond}[prop]{\bf Conditions}
\newenvironment{proof}
  {\begin{trivlist}\item[]{\bf Proof.}}
  {\hspace*{\fill}{$\bowtie$}\end{trivlist}}

\newcommand{\setR}{\mathds{R}}

\newcommand{\setN}{\mathds{N}}

\newcommand{\supIn}{\mathrm{in}}
\newcommand{\supOut}{\mathrm{out}}
\newcommand{\supLocal}{\mathrm{loc}}
\newcommand{\supGlobal}{\mathrm{glob}}
\newcommand{\supReturn}{\mathrm{return}}
\newcommand{\subSS}{\mathrm{ss}}
\newcommand{\subS}[1]{\mathrm{s#1}}
\newcommand{\subU}{\mathrm{u}}
\newcommand{\subC}{\mathrm{c}}
\newcommand{\subSSS}{{\mathrm{ss},\mathrm{s}}}
\newcommand{\subUS}{{\mathrm{u},\mathrm{s}}}
\newcommand{\SectionIn}{\Sigma^\supIn}
\newcommand{\SectionOut}{\Sigma^\supOut}
\newcommand{\SectionInInterior}{\Sigma^\supIn_{*}}
\newcommand{\SectionOutInterior}{\Sigma^\supOut_{*}}
\newcommand{\MapLocal}{\Psi^\supLocal}
\newcommand{\MapGlobal}{\Psi^\supGlobal}
\newcommand{\MapCombined}{\Psi}
\newcommand{\tloc}{t^\supLocal}
\newcommand{\muss}{\mu_\subSS}
\newcommand{\muu}{\mu_\subU}
\newcommand{\mus}[1]{\mu_\subS{#1}}
\newcommand{\xc}{x_\subC}
\newcommand{\xss}{x_\subSS}
\newcommand{\xu}{x_\subU}
\newcommand{\xs}[1]{x_\subS{#1}}
\newcommand{\xsss}{x_\subSSS}
\newcommand{\xus}{x_\subUS}
\newcommand{\xIn}{x^\supIn}
\newcommand{\xcIn}{x_\subC^\supIn}
\newcommand{\xssIn}{x_\subSS^\supIn}
\newcommand{\xuIn}{x_\subU^\supIn}
\newcommand{\xsIn}[1]{x_\subS{#1}^\supIn}
\newcommand{\xusIn}{x_\subUS^\supIn}
\newcommand{\xOut}{x^\supOut}
\newcommand{\xcOut}{x_\subC^\supOut}
\newcommand{\xssOut}{x_\subSS^\supOut}

\newcommand{\xsOut}[1]{x_\subS{#1}^\supOut}
\newcommand{\xsssOut}{x_\subSSS^\supOut}
\newcommand{\dx}{\delta}
\newcommand{\dxc}{\delta_\subC}
\newcommand{\dxss}{\delta_\subSS}
\newcommand{\dxu}{\delta_\subU}
\newcommand{\dxs}[1]{\delta_\subS{#1}}
\newcommand{\dxIn}{\delta^\supIn}
\newcommand{\dxcIn}{\delta_\subC^\supIn}
\newcommand{\dxssIn}{\delta_\subSS^\supIn}
\newcommand{\dxuIn}{\delta_\subU^\supIn}
\newcommand{\dxsIn}[1]{\delta_\subS{#1}^\supIn}

\newcommand{\dxcOut}{\delta_\subC^\supOut}
\newcommand{\dxssOut}{\delta_\subSS^\supOut}

\newcommand{\dxsOut}[1]{\delta_\subS{#1}^\supOut}
\newcommand{\gcss}{g_\mathrm{css}}
\newcommand{\gcs}[1]{g_\mathrm{cs#1}}
\newcommand{\gss}{g_\subSS}
\newcommand{\gs}[1]{g_\subS{#1}}
\newcommand{\gu}{g_\subU}
\newcommand{\neighborhood}{\mathcal{U}}
\newcommand{\diff}{\mathrm{\,d}}
\newcommand{\dist}{\mathrm{dist}}

\begin{document}
\pagestyle{empty}
\markboth
  {S.~Liebscher, A.D.~Rendall, S.B.~Tchapnda}
  {Oscillatory singularities in Bianchi models with magnetic fields}
\pagenumbering{arabic}

\vspace*{\fill}

\begin{center}
\LARGE\bfseries
Oscillatory singularities in \\
Bianchi models with magnetic fields
\end{center}

\vspace*{\fill}

\begin{center}
   \textbf{\large Stefan Liebscher}
\\ Freie Universit\"at Berlin, Institut f\"ur Mathematik
\\ Arnimallee 3, 14195 Berlin, Germany
\\ \texttt{stefan.liebscher@fu-berlin.de}
\\ ~
\\ \textbf{\large Alan D. Rendall}
\\ Max-Planck-Institut f\"ur Gravitationsphysik (Albert-Einstein-Institut)
\\ Am M\"uhlenberg 1, 14476 Potsdam, Germany
\\ \texttt{rendall@aei.mpg.de}
\\ ~
\\ \textbf{\large Sophonie Blaise Tchapnda}
\\ Department of Mathematics, University of Yaounde I
\\ P.O. Box 812, Yaounde, Cameroon
\\ \texttt{sophonieblaise@yahoo.com}
\end{center}

\vspace*{\fill}

\begin{center}
Preprint
\\
July 11, 2012 
\end{center}

\vspace*{\fill}

\clearpage

\vspace*{\fill}

\begin{abstract}\noindent
An idea which has been around in general relativity for more than forty years
is that in the approach to a big bang singularity solutions of the Einstein
equations can be approximated by the Kasner map, which
describes a succession of Kasner epochs. This is already a highly non-trivial
statement in the spatially homogeneous case. There the Einstein equations
reduce to ordinary differential equations and it becomes a statement that the 
solutions of the Einstein equations can be approximated by heteroclinic chains 
of the corresponding dynamical system. For a long time progress on 
proving a statement of this kind rigorously was very slow but recently there 
has been new progress in this area, particularly in the case of the vacuum 
Einstein equations. In this paper we generalize some of these results to the 
Einstein-Maxwell equations. It turns out that this requires new techniques 
since certain eigenvalues are in a less favourable configuration in the case 
with a magnetic field. The difficulties which arise in that case are
overcome by using the fact that the dynamical system of interest is of
geometrical origin and thus has useful invariant manifolds.   
\end{abstract}

\vspace*{\fill}

\cleardoublepage
\setcounter{page}{1}
\pagestyle{myheadings}

\section{Introduction}
\label{secIntroduction}

The fundamental equations of general relativity are the Einstein equations,
possibly coupled to other equations describing the dynamics of the matter which
generates the gravitational field. With a suitable choice of physical units 
the equations are
\begin{equation}
R_{\alpha\beta}-\frac12 Rg_{\alpha\beta}=T_{\alpha\beta}.
\end{equation}
The unknowns in these equations are the spacetime metric $g_{\alpha\beta}$ and 
the matter fields. $R_{\alpha\beta}$ is the Ricci tensor of the Lorentzian
metric $g_{\alpha\beta}$ and $R$ its trace. $T_{\alpha\beta}$ is the 
energy-momentum tensor. In this paper we are mainly concerned with the 
Einstein vacuum equations, where $T_{\alpha\beta}=0$, and the Einstein-Maxwell
equations. In the latter case the source of the gravitational field is
an electromagnetic field $F_{\alpha\beta}$ and the energy-momentum tensor is
given by
\begin{equation}
T_{\alpha\beta} \;=\;
  F_\alpha{}^\gamma F_{\beta\gamma}
  - \frac14 (F^{\gamma\delta}F_{\gamma\delta})g_{\alpha\beta}.
\end{equation}
The electromagnetic field tensor is antisymmetric 
($F_{\alpha\beta}=-F_{\beta\alpha}$) and satisfies the 
source-free Maxwell equations
\begin{equation}
\nabla^\alpha F_{\alpha\beta} \;=\; 0, \qquad
\nabla_\alpha F_{\beta\gamma} + \nabla_\gamma F_{\alpha\beta} 
  + \nabla_\beta F_{\gamma\alpha} \;=\; 0.
\end{equation}

It is well known that solutions of the Einstein equations generally develop
singularities. In particular, there are solutions relevant to cosmology in
which the singularity corresponds to the big bang. Belinskii, Khalatnikov
and Lifshitz (hereafter abbreviated to BKL) developed a heuristic picture
of the singularities in cosmological solutions of the Einstein equations. 
In this context they introduced a map of the circle to itself which we
refer to as the Kasner map. It will be defined precisely below. They
suggested that it provides a model for oscillations of the geometry in the
approach to the singularity. For the original work 
see \cite{BelinskiiKhalatnikovLifshitz1970-OscillatoryApproach} and 
\cite{BelinskiiKhalatnikovLifshitz1982-GeneralSolution}. 
A modern discussion of these ideas can be found in 
\cite{HeinzleUggla2009-Mixmaster}. 
An important idea in the BKL work is that spatially
inhomogeneous solutions of the Einstein equations can be approximated by
spatially homogeneous solutions near the singularity. Since from a 
mathematical point of view the dynamics of spatially homogeneous solutions
is still far from understood it is natural at the present time to concentrate
on understanding classes of spatially homogeneous solutions. This is the 
strategy we adopt in what follows.   

A long-standing question in mathematical cosmology is to relate the Kasner
map to the dynamics of actual solutions of the Einstein equations,
possibly with matter. An important recent advance in this field is the
paper \cite{LiebscherHaerterichWebsterGeorgi2011-BianchiA} 
where a relation of this kind was established 
in a special case. These results concern solutions of the vacuum 
Einstein equations of Bianchi types VIII and IX. They complement
earlier results of Ringstr\"om \cite{Ringstroem2000-CurvatureBlowup}, 
\cite{Ringstroem2001-BianchiIX} by 
providing a more detailed description of the dynamics of the approach to the 
singularity in certain cases. The work of Ringstr\"om on vacuum spacetimes was 
preceded by results of Weaver \cite{Weaver2000-MagneticBianchi}
on solutions of the 
Einstein-Maxwell equations of Bianchi type VI${}_0$ using a dynamical system
introduced in \cite{LeBlancKerrWainwright1995-MagneticBianchiVI}. 
The aim of this paper is to extend the results 
of \cite{LiebscherHaerterichWebsterGeorgi2011-BianchiA} 
to this case of the Einstein-Maxwell equations. 
There is other recent work on this question in the vacuum case 
\cite{Beguin2010-BianchiAsymptotics}, \cite{ReitererTrubowitz2010-BKL} 
but these papers use very different techniques from those 
which we will apply to the Einstein-Maxwell case and for this reason they will 
not be discussed further here.

In the next section the necessary background and the fundamental equations 
needed in the paper are introduced. The most important similarities and 
differences between the models with magnetic fields considered in what follows 
and the vacuum models which had previously been analysed are explained. The 
third section contains the main theorem and an exposition of the strategy of 
its proof. The central result is the existence of unstable manifolds of 
codimension one for some heteroclinic chains. To prove this it is necessary to
obtain estimates for a solution during its passages close to the Kasner
circle and for its behaviour between passages. This is done in sections 
\ref{secLocalPassage} and \ref{secReturnMap}, respectively. A central idea of 
the paper and one which is a major step
beyond what was achieved in the vacuum case is the use of a specially
constructed Riemannian metric to measure the distance between the heteroclinic
chains and the approximating smooth solutions. The last section discusses
future extensions of this research and interesting open problems.


\section{The basic set-up}
\label{secBasicSetup}

Spatially homogeneous spacetimes are those solutions of the Einstein-matter
equations where there is an action of a Lie group $G$ by 
isometries of $g_{\alpha\beta}$ with three-dimensional spacelike orbits which
leaves the matter fields invariant. The cases where the isotropy group is
discrete can be classified according to the Lie algebra of $G$. It is common 
in general relativity to use the terminology due to Bianchi, who 
introduced types I to IX. It is also common to distinguish between two
subsets of these types known as Class A and Class B. In what follows we
will only be concerned with Class A models. More information on this  
subject can be found in \cite{WainwrightEllis1997-Cosmology} or 
\cite{Rendall2008-Book}.

The analyses of vacuum spacetimes mentioned above are based on the well-known
Wainwright-Hsu system \cite{WainwrightHsu1989-BianchiA}. 
This is a system of ordinary
differential equations for five variables $(\Sigma_+,\Sigma_-,N_1,N_2,N_3)$
which are subject to one constraint. It includes all the Bianchi models of
Class A (i.e. types I, II, VI${}_0$, VII${}_0$,VIII and IX). The system is 
defined on a smooth hypersurface in $\setR^5$. An analogous system for Bianchi
spacetimes of type VI${}_0$ with a magnetic field was introduced in 
\cite{LeBlancKerrWainwright1995-MagneticBianchiVI}. 
It is also defined on a smooth hypersurface in $\setR^5$
and it includes solutions of types I and II with a magnetic field. The
variables are called $\Sigma_+,\Sigma_-,N_+,N_-,H$. The first two variables 
can be identified with the variables of the same name in the vacuum case
since they have the same geometrical meaning in both cases. The variables
$N_+$ and $N_-$ correspond in a similar way to certain linear combinations of
$N_2$ and $N_3$. More specifically, $N_+=\frac32 (N_2+N_3)$ and 
$N_-=\frac{\sqrt{3}}{2}(N_2-N_3)$. The variable $H$ corresponds to the 
magnetic field.

The dynamical system is
\begin{equation}\label{eqBianchiIIMagnetic}
\begin{array}{rcl}
\Sigma_+'&=&-2N_-^2(1+\Sigma_+)+\frac32 H^2(2-\Sigma_+),\\
\Sigma_-'&=&-(2N_-^2+\frac32 H^2)\Sigma_--2N_+N_-,\\
N_+'&=&(2\Sigma_+(1+\Sigma_+)+2\Sigma_-^2+\frac32 H^2)N_++6\Sigma_-N_-,\\
N_-'&=&(2\Sigma_+(1+\Sigma_+)+2\Sigma_-^2+\frac32 H^2)N_-+2\Sigma_-N_+,\\
H'&=&-(\Sigma_+(2-\Sigma_+)-\Sigma_-^2+N_-^2)H.
\end{array}
\end{equation}
The prime denotes a derivative with respect to a time variable $\tau$ which
tends to $-\infty$ as the singularity is approached.
These equations are taken from \cite{Weaver2000-MagneticBianchi}. 
They arise as a special case of the equations for models 
with a magnetic field and a perfect fluid derived in 
\cite{LeBlancKerrWainwright1995-MagneticBianchiVI} 
by setting the fluid density to zero. 
(Here a magnetic field means an electromagnetic field satisfying 
the condition that $F_{\alpha\beta}n^\beta=0$, where $n^\alpha$ 
is the unit normal vector to the group orbits.) 
Solutions are considered which satisfy the condition
\begin{equation}\label{eqBianchiIIMagneticConstraint}
\Sigma_+^2+\Sigma_-^2+N_-^2+\frac32 H^2 \;=\; 1.
\end{equation}
This condition follows from the Einstein equations and is preserved by the 
evolution equations for $(\Sigma_+,\Sigma_-,N_+,N_-,H)$ just defined. The 
inequalities $N_->0$, 
$N_+^2<3N_-^2$ and $H>0$ are assumed. These are also preserved by the evolution
and define the region which corresponds to Bianchi type VI${}_0$ solutions 
with non-zero magnetic field. Setting $H=0$ while maintaining the other 
two inequalities gives a representation of the vacuum solutions of Bianchi 
type VI${}_0$. Setting $N_-=0$ or $N_+=\sqrt{3}N_-$ gives two different
representations of solutions of type II with a magnetic field. Setting other
combinations to zero leads to vacuum solutions of type II, solutions of type I
with a magnetic field and vacuum solutions of type I (the Kasner solutions).
Note the invariant subspaces $\{N_2=0\}$ and $\{N_3=0\}$ which also 
appear in the Bianchi system with perfect fluid.
In fact the invariant subspaces $\{N_2=0\}$, $\{N_3=0\}$, $\{H=0\}$ 
will play a crucial role in our analysis, see (\ref{eqCondLocSubspaces}).

The circle defined by $\Sigma_+^2+\Sigma_-^2=1$ consists of stationary points.
Each one of them corresponds to a Kasner solution and so this set is called the
Kasner circle. There are three families of heteroclinic orbits between points on
the Kasner circle whose projections to the $(\Sigma_+,\Sigma_-)$-plane are 
straight lines. Two of these families correspond to vacuum solutions of 
Bianchi type II and occur in both the vacuum case and the case with magnetic 
field. In the vacuum case there is a third family related to these two by 
symmetries of the system. In the case where a magnetic field is included the 
two sets of Bianchi type II vacuum solutions are complemented by a family of 
Bianchi type I solutions with magnetic field. The projections of the latter 
to the  $(\Sigma_+,\Sigma_-)$-plane are identical to those of the third family 
of Bianchi type II solutions in the vacuum case. There is thus a natural 
correspondence between heteroclinic chains consisting of Bianchi type II
solutions in the vacuum case and heteroclinic chains in the case with a 
magnetic field which include orbits corresponding to both solutions 
of the vacuum Einstein equations of Bianchi type II and solutions of the  
Einstein-Maxwell equations of Bianchi type I. In the vacuum case there is a 
heteroclinic cycle consisting of three orbits and it is the central example
considered in \cite{LiebscherHaerterichWebsterGeorgi2011-BianchiA}. 
The projections of the orbits making up this 
cycle to the $(\Sigma_+,\Sigma_-)$-plane are related by rotations by multiples 
of $\frac{2\pi}{3}$. By what has already been said, there is a corresponding 
heteroclinic cycle in the system of 
\cite{LeBlancKerrWainwright1995-MagneticBianchiVI}. 
See also figure \ref{fig3cycle}.

\begin{figure}
\centering
\noindent%
\setlength{\unitlength}{0.2\textwidth}%
\begin{picture}(3.6,3.6)(-1.3,-1.8)
\put(0.5,0){\makebox(0,0){\includegraphics[width=3.6\unitlength]{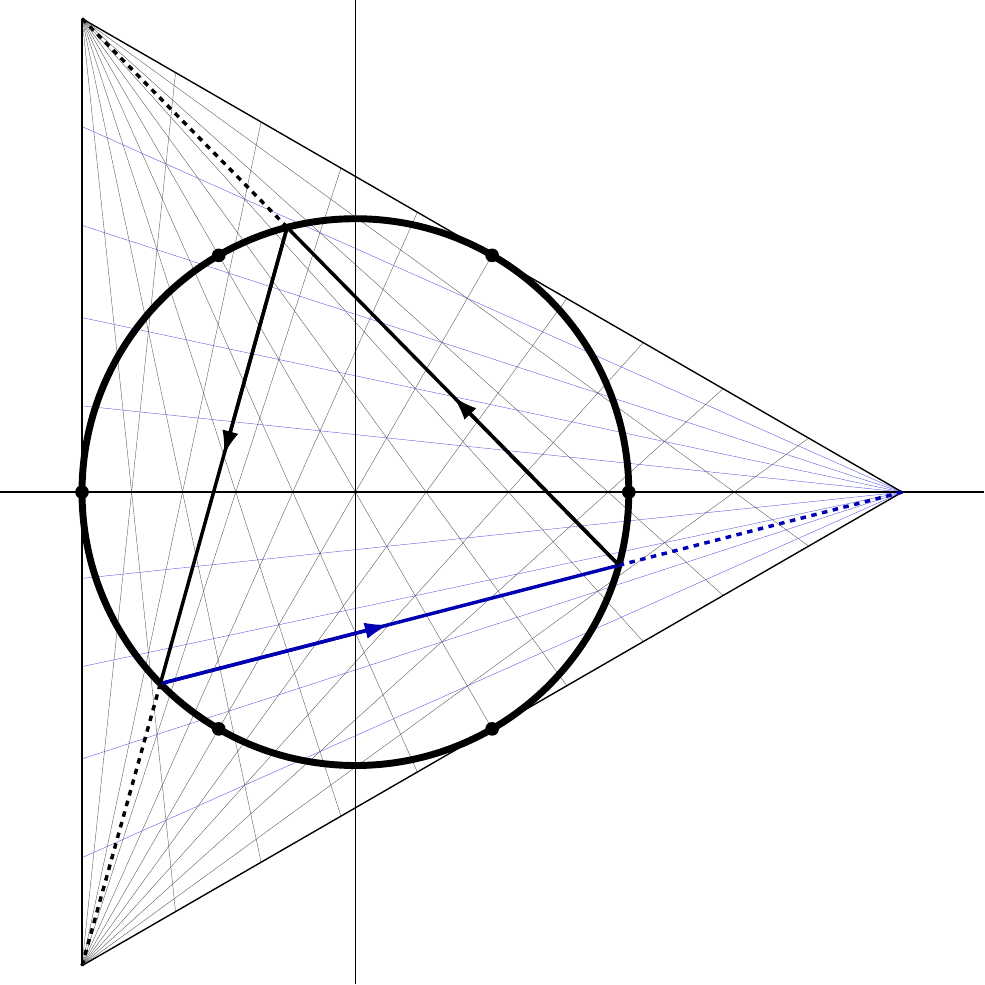}}}
\put( 2.3 ,-0.03){\makebox(0,0)[rt]{$\Sigma_+$}}
\put( 0.03, 1.8){\makebox(0,0)[lt]{$\Sigma_-$}}
\put( 0.91, 0.45){\makebox(0,0)[lb]{$\mathcal{K}_1$}}
\put( 0.0 , 1.03){\makebox(0,0)[b ]{$\mathcal{K}_2$}}
\put(-0.91, 0.45){\makebox(0,0)[rb]{$\mathcal{K}_3$}}
\put(-0.91,-0.45){\makebox(0,0)[rt]{$\mathcal{K}_4$}}
\put( 0.0 ,-1.03){\makebox(0,0)[t ]{$\mathcal{K}_5$}}
\put( 0.91,-0.45){\makebox(0,0)[lt]{$\mathcal{K}_6$}}
\put(-1.03,-0.03){\makebox(0,0)[rt]{$T_1$}}
\put( 0.53, 0.88){\makebox(0,0)[lb]{$T_2$}}
\put( 0.53,-0.88){\makebox(0,0)[lt]{$T_3$}}
\put( 1.03, 0.03){\makebox(0,0)[lb]{$Q_1$}}
\put(-0.53,-0.88){\makebox(0,0)[rt]{$Q_2$}}
\put(-0.53, 0.88){\makebox(0,0)[rb]{$Q_3$}}
\end{picture}
\caption{\label{fig3cycle}
  Kasner circle and period 3 cycle;
  Bianchi type II vacuum families in black, 
  Bianchi type I family with magnetic field in blue;
  Kasner intervals $\mathcal{K}_i$ bounded by 
  Taub points $T_i$ of tangencies and antipodal points $Q_i$.}
\end{figure}
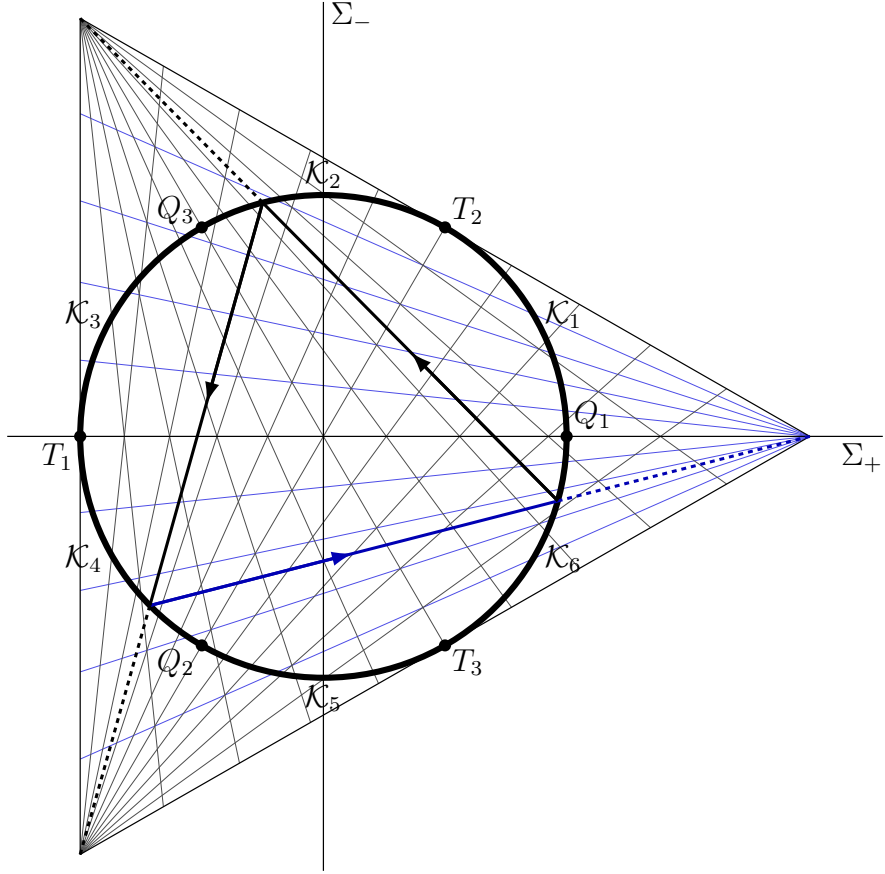

The Kasner solutions can be written in the explicit form
\begin{equation}
-dt^2+t^{2p_1}dx^2+t^{2p_2}dy^2+t^{2p_3}dz^2.
\end{equation}
With a suitable choice of ordering the Kasner exponents $p_1$, $p_2$, $p_3$ 
are related to the variables $\Sigma_+$ and $\Sigma_-$  by
\begin{equation}
\begin{array}{rcl}
p_1&=&\frac13 (1-2\Sigma_+),\\
p_2&=&\frac13 (1+\Sigma_++\sqrt{3}\Sigma_-),\\
p_3&=&\frac13 (1+\Sigma_+-\sqrt{3}\Sigma_-).
\end{array}
\end{equation}
The eigenvalues of the linearisation of the Wainwright-Hsu system at a Kasner 
solution are $(6p_1,6p_2,6p_3)$ (\cite{WainwrightHsu1989-BianchiA}, p. 1425).
For the system describing solutions of the Einstein-Maxwell equations 
of Bianchi type VI${}_0$ introduced in 
\cite{LeBlancKerrWainwright1995-MagneticBianchiVI} 
the eigenvalues are $(3p_1,6p_2,6p_3)$ 
(\cite{LeBlancKerrWainwright1995-MagneticBianchiVI}, (4.1)). 
The Taub points of the Kasner circle are defined by 
the condition that the Kasner exponents are $(1,0,0)$ or a permutation thereof.
All points of the Kasner circle other than the Taub points have a 
one-dimensional stable manifold. For $p_1<0$ the stable manifold is defined by 
a solution with non-vanishing magnetic field and the corresponding eigenvalue 
is $3p_1<0$. For $p_2<0$ or $p_3<0$ the stable manifold is defined by a vacuum 
solution.

In order to have an overview of the relative sizes of the different eigenvalues
it is useful to introduce the Kasner parameter $u\in [1,\infty]$ which is 
defined implicitly by the relation (cf. \cite{HeinzleUggla2009-Mixmaster})
\begin{equation}
p_1p_2p_3=\frac{-u^2(1+u)^2}{(1+u+u^2)^3}.
\end{equation}
Then the Kasner exponents arranged in ascending order are given by
\begin{equation}
\begin{array}{rcl}
\tilde p_1&=&\frac{-u}{1+u+u^2},\\
\tilde p_2&=&\frac{1+u}{1+u+u^2},\\
\tilde p_3&=&\frac{u(1+u)}{1+u+u^2}.
\end{array}
\end{equation}
Note that $\tilde p_2\ge -\tilde p_1$ with equality only when $u=\infty$. On 
the other hand $\frac{\tilde p_2}{2}\le -\tilde p_1$ and 
$\frac{\tilde p_3}{2}\ge -\tilde p_1$. The Kasner
map is defined by $u\mapsto u-1$ for $u\ge 2$ and $u\mapsto (u-1)^{-1}$ 
for $u\le 2$. 
The central example in \cite{LiebscherHaerterichWebsterGeorgi2011-BianchiA}
is a set of heteroclinic 
orbits which form a cycle of order three. In fact the value of $u$ 
corresponding to that example is invariant under the Kasner map. It 
solves the 
equation $u^2-u-1=0$ and is the golden ratio $\frac12(1+\sqrt{5})$. The values 
of the Kasner exponents at the vertex of this cycle where $p_1<p_2<p_3$ are
\begin{equation}\label{eq3cyclePvalues}
\begin{array}{rcl}
p_1&=&\frac14 (1-\sqrt{5}),\\
p_2&=&\frac12,\\
p_3&=&\frac14 (1+\sqrt{5}).
\end{array}
\end{equation}

For each value of the Kasner parameter $u$ in the interval $(1,\infty)$ there
are six points on the Kasner circle where $u$ takes that value. Removing
the Taub points $T_i$ and their antipodal points $Q_i$ from the Kasner 
circle leaves a union of six intervals $K_i$, $1\le i\le 6$. They will be
numbered as follows. Let $K_1$ be the region where $p_1<0$ and $p_2>p_3$.
Then number the others consecutively while moving anticlockwise along the 
Kasner circle, see figure \ref{fig3cycle}.

In order to assess the stability of a heteroclinic cycle it is important
to examine the eigenvalues of the linearisation of the system at the 
vertices. In both the vacuum case and the case with magnetic field there is 
one negative eigenvalue $-\mu_1$ and two positive eigenvalues $\mu_2$, 
$\mu_3$. Without loss of generality the labelling can be chosen so that
$\mu_2\le\mu_3$. Then from what has been stated above it can be seen that in 
the vacuum case the inequalities $-\mu_1<\mu_2<\mu_3$ hold at any point of the 
Kasner circle except $T_i$ and $Q_i$. Call this the first linearisation 
condition. That this is true is one of the most important hypotheses of 
the main theorem of \cite{LiebscherHaerterichWebsterGeorgi2011-BianchiA}. 
On the other hand this condition can fail in the case with magnetic field.
It fails precisely when the Kasner exponent $p_1$ is intermediate in size
between $p_2$ and $p_3$, i.e. when the point of the Kasner circle lies in the
one of the sets $K_2$ and $K_5$. For then the eigenvalue $3p_1$ is smaller in
magnitude than that of the negative eigenvalue, while the other positive 
eigenvalue is not. Then we have the situation that $\mu_2<-\mu_1<\mu_3$. 
The situation that the eigenvalue $3p_1$ does not correspond to the 
eigendirection tangent to the heteroclinic orbit incoming towards the past is 
covered by the theorems in this paper. Call this the second linearisation 
condition. What is common to the first and second linearisation conditions is
that the eigenvalue corresponding to the heteroclinic orbit incoming towards
the past is larger in modulus that that corresponding to the heteroclinic
orbit outgoing towards the past. In the example of the 3-cycle at least one of 
the two linearisation conditions just introduced holds at each of the the 
vertices. 
See figure \ref{fig3cycle}, the second eigenvalue condition holds 
in the intervals $\mathcal{K}_2$ and $\mathcal{K}_5$, whereas 
the first eigenvalue condition holds in the remaining intervals.

The linearisation conditions are not in themselves enough to make the 
theorems in this paper work. Additional geometrical information is required. 
This is the existence of a certain invariant manifold. It is tangent to the 
space spanned by the vectors tangent to the stable manifold and the centre 
manifold and the eigenvector corresponding to the largest eigenvalue. For a 
general dynamical system there is no reason why a manifold of this kind should 
exist. In the example of a Bianchi model of type VI${}_0$ with magnetic field 
a manifold of this kind is defined by the vacuum solutions of type VI${}_0$
or the solutions of type II with magnetic field.


\section{Main result and sketch of proof}
\label{secMainResult}

We shall prove the following result on the dynamics 
of the Bianchi model of type VI${}_0$ 
(\ref{eqBianchiIIMagnetic}, \ref{eqBianchiIIMagneticConstraint})
with magnetic field.

\begin{thm}\label{thPeriod3Manifold}
The period 3 heteroclinic cycle given by (\ref{eq3cyclePvalues}) possesses a 
local codimen\-sion-one unstable manifold.
In other words, 
system (\ref{eqBianchiIIMagnetic}, \ref{eqBianchiIIMagneticConstraint}) 
admits a codimension-one manifold, locally close to the heteroclinic cycle,
of initial conditions whose backward trajectories converge 
to the heteroclinic cycle.
The manifold is locally Lipschitz continuous in the open complement of
the boundaries ${N_1=0}$, ${N_2=0}$, ${H=0}$. 
It is Lipschitz continuous in every closed cone 
intersecting these boundaries only in the heteroclinic cycle.
\end{thm}

The proof will only use certain properties of the particular structure of the
Bianchi system (\ref{eqBianchiIIMagnetic}) and can be sketched as follows.

\textbf{Step 1: local passage}, section \ref{secLocalPassage}.
In a neighbourhood of the equilibria of the heteroclinic cycle, 
i.e. close to the Kasner circle, the Bianchi system 
(\ref{eqBianchiIIMagnetic}, \ref{eqBianchiIIMagneticConstraint}) 
with reversed time direction can be smoothly transformed 
to a vector field 
\[
\dot{x}=f(x), \quad
x=(\xu,\xss,\xs{},\xc)
  \in\setR\times\setR\times\setR^N\times\setR,
\]
that satisfies the following properties:

\begin{cond}\label{condLocalPassage}
\begin{description}[itemsep=0ex]
\item[(loc-i)] 
There is a straight line of equilibria,
\begin{equation}\label{eqCondLocEquilibria}
f(0,0,0,\xc) \equiv 0.
\end{equation}
\item[(loc-ii)] 
The heteroclinic orbits of the original system 
correspond to the $\xss$- and $\xu$-axes.
\item[(loc-iii)]
The linearisation at the origin has the almost diagonal form
\begin{equation}\label{eqCondLocLinearization}
Df(0) = \left(\begin{array}{cccccc}
  \muu\\
  &-\muss\\
  &&-\mus{1}\\
  &&&\ddots\\
  &&&&-\mus{N}\\
  {}*&*&*&\cdots&*&0
\end{array}\right),
\end{equation}
with $\muu, \muss, \mus{1}, \ldots, \mus{N} > 0$.
\item[(loc-iv)] 
The eigenvalue corresponding to the incoming direction 
is stronger than the eigenvalue corresponding to the outgoing direction, 
\begin{equation}\label{eqCondLocEigenvalues}
\muu / \muss < 1.
\end{equation}
\item[(loc-v)] 
The codimension-one subspaces
\begin{equation}\label{eqCondLocSubspaces}
\{\xu = 0\}, \quad \{\xss = 0\}, \quad \{\xs{1} = 0\}, \ldots, \{\xs{N} = 0\}
\end{equation}
are invariant under the flow.
\end{description}
\end{cond}

Note that, for the model
(\ref{eqBianchiIIMagnetic}, \ref{eqBianchiIIMagneticConstraint}),
we have $N=1$.
However the generalization $N>1$ is needed for applications discussed at the 
end of this section and in section \ref{secDiscussion}.

The last two properties are the most crucial for the proof. 
The invariance of the codimension-one subspaces is non-generic 
for systems admitting (loc-i)--(loc-iv) and is a very strong constraint
on the system.

Under the above assumptions the local passage map, 
$\MapLocal:\SectionIn\to\SectionOut$, 
from an in-section $\SectionIn=\{\xss=\varepsilon\}$
to an out-section $\SectionOut=\{\xu=\varepsilon\}$,
for sufficiently small $\varepsilon$, is a Lipschitz-continuous map 
with arbitrarily little change of the $\xc$-component and arbitrarily 
strong contraction transverse to the $\xc$-component.
Unfortunately, this only holds true with respect to a non-Euclidean metric
and represents one of the main difficulties of our investigation.
Indeed, for $\mus{1}<\muu$ even a linear vector field $f$ never gives rise 
to a Lipschitz-continuous map $\MapLocal$ with respect 
to the Euclidean metric.  
The case $\muu < \muss,\mus{1},\ldots\mus{N}$, on the other hand, 
does yield a Lipschitz-continuous map $\MapLocal$ and has been treated in 
\cite{LiebscherHaerterichWebsterGeorgi2011-BianchiA}.

\textbf{Step 2: global excursion}, section \ref{secReturnMap}.
Close to the heteroclinic chain, 
by smooth dependence on initial conditions,
the trajectories follow the heteroclinic orbit from the out-section of a
local passage to the in-section of the next local passage.
This map, $\MapGlobal_k:\SectionOut_k\to\SectionIn_{k+1}$,
is a uniformly bounded diffeomorphism.
In particular, any deformation imposed by $\MapGlobal$ in directions transverse 
to $\xc$ will turn out to be dominated by the strong contraction of the local 
passage map $\MapLocal$.
In $\xc$-direction, however, we gain an expansion given by the Kasner map.
Thus, the global excursion $\MapGlobal_k:\SectionOut_k\to\SectionIn_{k+1}$
given by the Bianchi system 
(\ref{eqBianchiIIMagnetic}, \ref{eqBianchiIIMagneticConstraint}) 
with reversed time direction satisfies

\begin{cond}\label{condGlobalExcursion}
\begin{description}[itemsep=0ex]
\item[(glob-i)] 
$\MapGlobal_k$ maps the origin of $\SectionOut_k$ to the origin of 
$\SectionIn_{k+1}$ and (local neighbourhoods of 0 of) the invariant subspaces
$\{\xss = 0\}$, $\{\xs{1} = 0\}, \ldots, \{\xs{N} = 0\}$ onto 
$\{\xu = 0\}$, $\{\xs{1} = 0\}, \ldots, \{\xs{N} = 0\}$ (in arbitrary order).
\item[(glob-ii)]
$\MapGlobal_k$ is a $\mathcal{C}^2$ Diffeomorphism. The bounds
$\|D\MapGlobal_k\|$, $\|D(\MapGlobal_k)^{-1}\|$, $\|D^2\MapGlobal_k\|$, 
$\|D^2(\MapGlobal_k)^{-1}\| < M$ are independent of $k$.
\item[(glob-iii)]
It uniformly expands in $\xc$-direction at the boundary. 
In other words,
$\MapGlobal_k:\{\xssOut = 0, \xsOut{} = 0\} \to \{\xuIn = 0, \xsIn{} = 0\}$
is Lipschitz continuous and its inverse 
$(\MapGlobal_k|_{\{\xssOut = 0, \xsOut{} = 0\}})^{-1}$ has Lipschitz constant 
less than $L<1$, independent of $k$.
\end{description}
\end{cond}

\textbf{Step 3: graph transform}, section \ref{secReturnMap}.
Combining local passage and global excursion yields maps 
from each in-section to the next,
$\MapCombined=\MapGlobal_k\circ\MapLocal_k:\SectionIn_k\to\SectionIn_{k+1}$,
with uniform cone conditions.
A standard graph-transform technique now yields the claimed invariant manifold 
as a fixed point in the space of Lipschitz-continuous graphs
$\xc^k=\xc^k(\xu^k, \xs{}^k)$ in $\SectionIn_k$.
For completeness of presentation, we will give the necessary 
arguments in section \ref{secReturnMap}.

In fact, steps 1--3 prove a much more general theorem 
than \ref{thPeriod3Manifold}, that is:

\begin{thm}\label{thAbstractSequenceManifold}
Let a $\mathcal{C}^4$ vector field and a chain of heteroclinic orbits $h^k(t)$,
\[
\lim_{t\to\infty} h^{k-1}(t) = p^k = \lim_{t\to-\infty} h^k(t), \qquad k\in\setN.
\]
be given. Assume that locally near $p^k$ assumptions (loc-i)--(loc-v) hold and 
that along each $h^k$ assumptions (glob-i)--(glob-iii) hold, 
with constants $\alpha, L, M$ independent of $k$.

Then there exists a local codimension-one stable manifold 
to the heteroclinic chain, 
i.e. a codimension-one manifold of initial conditions following 
the heteroclinic chain and converging to it.

The heteroclinic chain itself is contained in the boundary of this manifold.
The manifold is locally Lipschitz continuous in the open complement of the
invariant subspaces (loc-v). 
The manifold is uniformly Lipschitz continuous in every closed cone 
intersecting the invariant subspaces only in the heteroclinic chain itself.
\end{thm}

This theorem covers not only the period 3 cycle of the Bianchi VI${}_0$ system
with magnetic field. 
In fact, it applies to every heteroclinic chain in the Bianchi VI${}_0$ system 
with magnetic field that does not accumulate at any Taub point 
(required for uniformity of bounds) and such that the chain does not contain
heteroclinic orbits of the magnetic family to points in the 
domain $K_2\cup K_5$. See also figure \ref{fig3cycle}.
It also applies to every heteroclinic chain of the Bianchi A (VIII and IX) 
system without magnetic field but with ideal fluid 
as investigated in \cite{LiebscherHaerterichWebsterGeorgi2011-BianchiA},
as long as it does not accumulate at any Taub point. 
The proof given here completes the arguments sketched in the discussion 
section of \cite{LiebscherHaerterichWebsterGeorgi2011-BianchiA} 
and relaxes the constraint on the matter model required there. 
All matter models that yield positive eigenvalues of the linearisation at 
the Kasner circle in the non-vacuum direction are included, 
due to the relaxed eigenvalue condition.


\section{Local passage near a line of equilibria}
\label{secLocalPassage}

In this section we study the passage of trajectories under a general flow 
near a line of equilibria with eigenvalue constraint (\ref{eqCondLocEigenvalues})
and invariant subspaces (\ref{eqCondLocSubspaces}) 
consistent with the Kasner circle 
in the Bianchi VI${}_0$ system with magnetic field.
We will collect estimates on expansion and contraction rates 
to establish Lipschitz properties of the local map between sections 
to a reference orbit given by the passage near the line of equilibria, 
see theorem \ref{thLocalMapLipschitz} at the end of this section.
Compared to \cite{LiebscherHaerterichWebsterGeorgi2011-BianchiA}(section 3),
we assume the relaxed eigenvalue condition (\ref{eqCondLocEigenvalues})
without any constraint on $\mus{1},\ldots,\mus{N}$. 
This requires the use of a non-Euclidean
metric (\ref{eqMetricIn}, \ref{eqMetricOut}).

Consider a $\mathcal{C}^k$ vector field, $k\ge 4$,
\begin{equation}\label{eqGeneralLocalVectorField}
\dot{x} \;=\; f(x), \quad
x \;=\; (\xu,\xss,\xs{},\xc)
  \;\in\; \setR\times\setR\times\setR^N\times\setR,
\end{equation}
that satisfies conditions \ref{condLocalPassage} 
in a neighbourhood of the origin. 
Due to the invariant subspaces (\ref{eqCondLocSubspaces}), 
the form of the linearisation (\ref{eqCondLocLinearization}) holds 
locally all along the line of equilibria,
\begin{equation}\label{eqGeneralCrudeLinearization}
Df(0,0,0,\xc) \;=\; \left(\begin{array}{cccccc}
  \muu(\xc)\\
  &-\muss(\xc)\\
  &&-\mus{1(\xc)}\\
  &&&\ddots\\
  &&&&-\mus{N}(\xc)\\
  {}*&*&*&\cdots&*&0
\end{array}\right).
\end{equation}

The stable and unstable manifolds as well as 
the strong stable foliation of the stable manifold 
are $\mathcal{C}^k$ and can be flattened, 
see e.g. \cite{ShilnikovTuraevChua1998-MethodsNonlinearDynamicsI}, Theorem 5.8.
By a $\mathcal{C}^k$ change of coordinates 
the stable / strong stable / unstable manifolds to the 
equilibria locally coincide with the respective eigenspaces.
In particular, and in addition to (\ref{eqCondLocSubspaces}), 
the following stable and unstable fibres become invariant:
\begin{equation}\label{eqGeneralInvManifolds}
\begin{array}{rcl}
W^\mathrm{u}(\xc) &=& \{\xss = 0, \xs{}=0, \xc \textrm{ fixed} \}, \\
W^\mathrm{s}(\xc) &=& \{\xu = 0, \xc \textrm{ fixed} \}.
\end{array}
\end{equation}
Note that in the Bianchi system, $W^\mathrm{u}(\xc)$ coincides with the 
outgoing heteroclinic orbit attached to the equilibrium $(0,0,0,\xc)$.

Due to (\ref{eqGeneralInvManifolds}), the linearisation becomes diagonal,
\begin{equation}\label{eqGeneralDiagonalLinearization}
Df(0,0,0,\xc) \;=\; \left(\begin{array}{cccccc}
  \muu(\xc)\\
  &-\muss(\xc)\\
  &&-\mus{1(\xc)}\\
  &&&\ddots\\
  &&&&-\mus{N}(\xc)\\
  0&0&0&\cdots&0&0
\end{array}\right).
\end{equation}

Our aim is to study a local map from an in-section 
$\SectionIn=\{\xss=\varepsilon\}$ to an 
out-section $\SectionOut=\{\xu=\varepsilon\}$ for $\xs{}, \xc \approx 0$, 
see figure \ref{figLocalMap}.
This corresponds to the passage near the Kasner circle in the Bianchi system 
in backwards time direction. 
(We reversed the time direction to obtain a well defined local map.)

\begin{figure}
\setlength{\unitlength}{0.01125\textwidth}
\centering
\begin{picture}(43,30)(-8,-5)
\put(-8,-5){\makebox(0,0)[bl]{\includegraphics[width=43\unitlength]{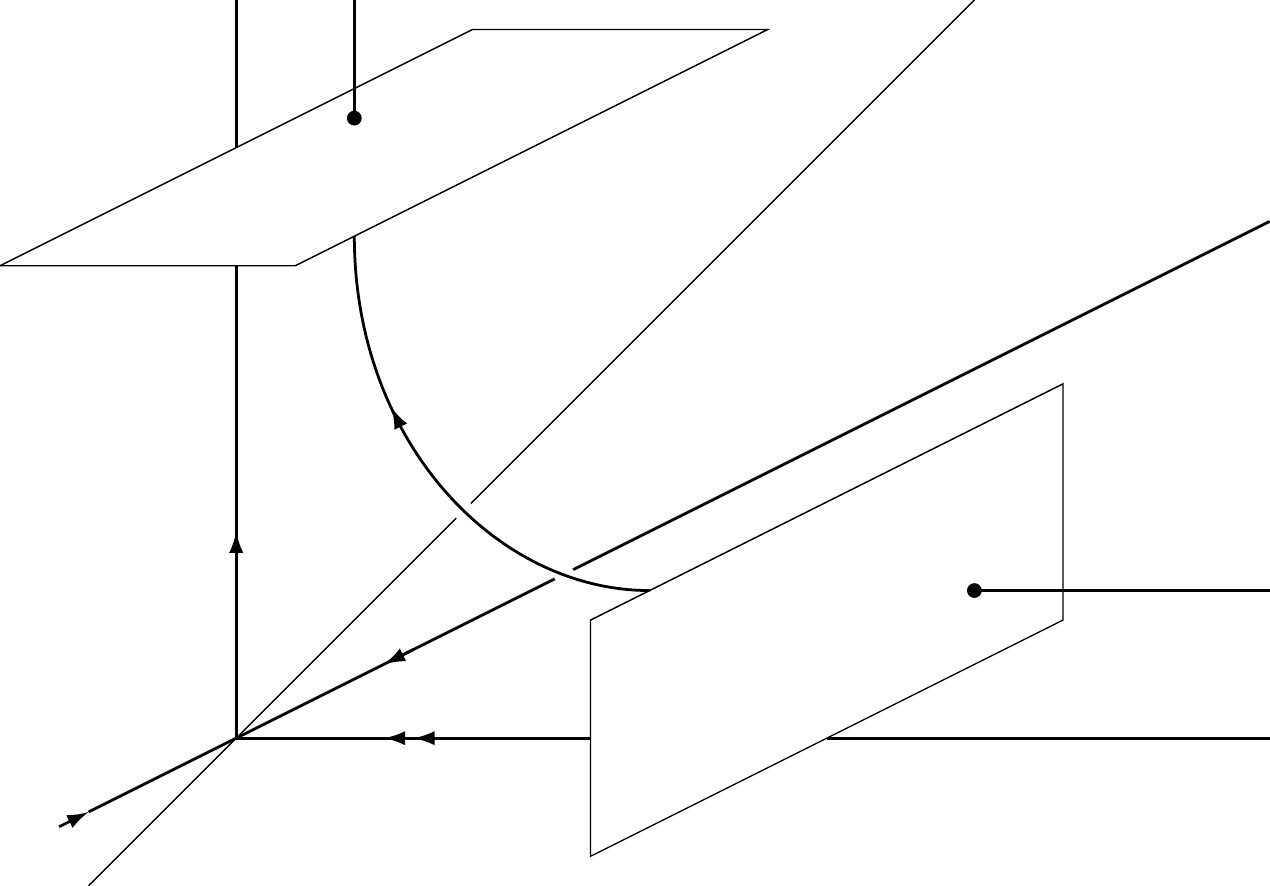}}}
\put(35,-0.5){\makebox(0,0)[rt]{$\xss$}}
\put(-0.5,25){\makebox(0,0)[rt]{$\xu$}}
\put(35,17.5){\makebox(0,0)[rb]{$\xs{}$}}
\put(25,25){\makebox(0,0)[lt]{$\xc$}}
\put(28.5, 8){\makebox(0,0)[l]{$\SectionIn$}}
\put(13,24.5){\makebox(0,0)[b]{$\SectionOut$}}
\put(25,5){\makebox(0,0)[br]{$\xIn$}}
\put(4,21){\makebox(0,0)[tl]{$\xOut$}}
\end{picture}
\caption{\label{figLocalMap}Local passage $\MapLocal:\SectionIn\to\SectionOut$.}
\end{figure}

We rescale the system to
\begin{equation}\label{eqGeneralLocalVectorFieldLinNonlin}
\dot{x} \;=\; A(\xc)x + \varepsilon g(x),
\end{equation}
with $\varepsilon$ arbitrarily fixed and $g$ at least quadratic 
in $(\xu, \xss, \xs{})$.
Due to the invariant subspaces (\ref{eqCondLocSubspaces}) and 
(\ref{eqGeneralInvManifolds}), the vector field takes the form
\begin{equation}\label{eqGeneralLocalVectorFieldRescaled}
\begin{array}{lcrcll}
\dot\xu       &=& \muu(\xc)\xu\,            &+& \varepsilon \gu(x)\xu, \\
\dot\xss      &=& -\muss(\xc)\xss           &+& \varepsilon \gss(x)\xss, \\
\dot\xs{\ell} &=& -\mus{\ell}(\xc)\xs{\ell} &+& \varepsilon \gs{\ell}(x)\xs{\ell}, 
  & \ell=1,\ldots N, \\
\dot\xc    &=& && \varepsilon \left( 
  \gcss(x)\xss + \sum_{\ell=1}^N \gcs{\ell}(x)\xs{\ell} \right) \xu,
\end{array}
\end{equation}
with $\mathcal{C}^{k-1}$-functions $\gu$, $\gss, \gs{\ell}$, 
vanishing along the line of equilibria, and 
$\mathcal{C}^{k-2}$-functions $\gcss, \gcs{\ell}$.
In particular
\begin{equation}\label{eqGeneralHotBounds}
\begin{array}{rcl}
|\gu(x)|, |\gss(x)|, |\gs{1}(x)|, \ldots, |\gs{N}(x)|
  & < & C \max(|\xu|,|\xss|,|\xs{1}|,\ldots,|\xs{N}|), \\
|\gcss(x)|, |\gcs{1}(x)|, \ldots, |\gcs{N}(x)| & < & C,
\end{array}
\end{equation}
for some constant $C>0$ independent of $\varepsilon$ and $x\in\neighborhood$, 
where $\neighborhood$ is some local neighbourhood of the origin. 
Similarly, all derivatives of $\gu,\gss, \gs{\ell}, \gcss, \gcs{\ell}$ 
are bounded by $C$ for $x\in\neighborhood$.
We choose
\begin{equation}\label{eqLocalNeighborhood}
\neighborhood=(-2,2)^{N+3}.
\end{equation}
All further estimates will use this rescaled system 
(\ref{eqGeneralLocalVectorFieldLinNonlin}) 
with flattened invariant manifolds (\ref{eqGeneralInvManifolds}) 
in the local neighbourhood $\neighborhood$.
They will be valid for all $\varepsilon<\varepsilon_0$ 
and suitably chosen $\varepsilon_0$. 
In the original system (\ref{eqGeneralLocalVectorField}),
$\varepsilon_0$ bounds the size of the neighbourhood 
of the origin in which this local analysis is valid.

\begin{prop}\label{thEigenvalueBounds}
Let
\[
\muu := \muu(0), \quad -\muss := -\muss(0), \quad 
-\mus{\ell} := -\mus{\ell}(0), \qquad \ell = 1,\ldots,N,
\]
be the eigenvalues of (\ref{eqGeneralDiagonalLinearization}) at the origin.
Then for all $0<\alpha<1$ there exists an $\varepsilon_0>0$ such that 
for all $\varepsilon<\varepsilon_0$ in 
(\ref{eqGeneralLocalVectorFieldLinNonlin}) and $x\in\neighborhood$
\begin{equation}\label{eqEigenvalueBounds}
\alpha \;\le\; 
\frac{\muu(\xc)}{\mu_u}, \frac{\muss(\xc)}{\muss}, 
  \frac{\mus{1}(\xc)}{\mus{1}}, \ldots, \frac{\mus{N}(\xc)}{\mus{N}} 
\;\le\; \alpha^{-1}.
\end{equation}
\end{prop}
\begin{proof}
Due to the invariant subspaces 
(\ref{eqCondLocSubspaces}, \ref{eqGeneralInvManifolds}),
the linearisation of the system at eigenvalues close to the origin remains diagonal,
and the eigenvalues depend differentiably on $\xc$, 
For the rescaled system (\ref{eqGeneralLocalVectorFieldLinNonlin}) 
with small $\varepsilon_0$ this provides bounds in $\neighborhood$: 
Indeed, there exists a constant $C>0$ independent of 
$\varepsilon_0$, $\varepsilon$, such that
\begin{equation}\label{eqEigenvalueGradients}
\textstyle
\left| \frac{\diff}{\diff \xc}\muu(\xc) \right|,\;
\left| \frac{\diff}{\diff \xc}\muss(\xc) \right|,\;
\left| \frac{\diff}{\diff \xc}\mus{\ell}(\xc) \right| \;<\; \varepsilon C.
\end{equation}
\end{proof}

The scalar function 
$\theta(x):= \muu \left(\muu(\xc) + \varepsilon\gu(x)\right)^{-1}$ 
is therefore $\mathcal{C}^{k-1}$ and close to $1$.
The vector field 
\[
  x' \;=\; \theta(x)f(x) \;=\; \frac{\muu}{\muu(\xc) + \varepsilon\gu(x)} f(x)
\]
has the same trajectories as the original vector field 
and all previous considerations remain valid. 
Thus we can assume, without loss of generality, 
that $\theta(x) \equiv 1$ in $\neighborhood$, i.e.
\begin{equation}\label{eqEulerMultiplier}
 \muu(\xc) \;\equiv\; \muu, \qquad \gu(x) \equiv 0.
\end{equation}
At this step we have made use of the fact that the origin possesses exactly 
one unstable eigenvalue. 
The vector field to consider then has the form
\begin{equation}\label{eqSplitLocalSystem}
\begin{array}{lcrcll}
\dot\xu       &=& \muu\xu \\
\dot\xss      &=& -\muss(\xc)\xss           &+& \varepsilon \gss(x)\xss, \\
\dot\xs{\ell} &=& -\mus{\ell}(\xc)\xs{\ell} &+& \varepsilon \gs{\ell}(x)\xs{\ell}, 
  & \ell=1,\ldots N, \\
\dot\xc    &=& && \varepsilon \left( 
  \gcss(x)\xss + \sum_{\ell=1}^N \gcs{\ell}(x)\xs{\ell} \right) \xu.
\end{array}
\end{equation}

\begin{lem}\label{thLocalBounds}
For all $0<\alpha<1$ there exists an $\varepsilon_0>0$ such that 
for all $\varepsilon<\varepsilon_0$, $x(0)\in\neighborhood$ and $t\ge 0$, 
as long as $x(t)$ remains in $\neighborhood$ under the flow to 
the vector field (\ref{eqSplitLocalSystem}),
we can estimate:
\begin{eqnarray}\label{eqLocalUnstableBound}
\xu(t) &=& \exp(\muu t) \, \xu(0),
\\\label{eqLocalStrongStableBound} 
\xss(t) &\in& \textstyle
  \left[ \exp\left(-\frac{1}{\alpha}\muss t\right),
    \exp\left(-\alpha\muss t\right) \right] \xss(0),
\\\label{eqLocalStableBound} 
\xs{\ell}(t) &\in& \textstyle
  \left[ \exp\left(-\frac{1}{\alpha}\mus{\ell} t\right),
    \exp\left(-\alpha\mus{\ell} t\right) \right] \xs{\ell}(0),
  \qquad \ell=1,\ldots,N,
\\\label{eqLocalCenterBound}
|\xc(t) - \xc(0)| &\le& \frac{2\varepsilon C}{\alpha} \left( 
  \frac{1}{\muss-\muu}|\xu(0)| + 
  \sum_{\ell=1}^N \frac{1}{\mus{\ell}}|\xs{\ell}(0)| \right).
\end{eqnarray}
Here, $C$ is the uniform (in $x$ and $\varepsilon$) bound 
from (\ref{eqGeneralHotBounds}).
\end{lem}
\begin{proof}
The unstable component (\ref{eqLocalUnstableBound}) is given directly 
by the vector field.
The estimates of the stable components 
(\ref{eqLocalStrongStableBound}, \ref{eqLocalStableBound})
follow from the vector field and the uniform bounds 
(\ref{eqEigenvalueBounds},  \ref{eqGeneralHotBounds}).
Indeed for arbitrary $0<\tilde{\alpha}<1$ and $\varepsilon < \varepsilon_0$
small enough, we have
\[
\dot\xss \;\in\; \left[-\frac{1}{\tilde\alpha}\muss - \varepsilon C,
  -\tilde\alpha\muss + \varepsilon C\right] \xss.
\]
Thus for arbitrary $0<\alpha<1$ we find suitable 
$\alpha < \tilde\alpha < 1$ and $\varepsilon_0$ small enough such that
\[
\dot\xss \;\in\; \left[ -\frac{1}{\alpha}\muss, -\alpha\muss \right] \xss.
\]
Integration yields the claim. Bounds on $\xs{\ell}$ are obtained analogously.
The centre component (\ref{eqLocalCenterBound}) is then estimated by plugging
(\ref{eqLocalUnstableBound}, \ref{eqLocalStrongStableBound}, 
\ref{eqLocalStableBound}) into the vector field (\ref{eqSplitLocalSystem}) 
and integrating:
\[
\begin{array}{rcl} 
\displaystyle |\xc(t) - \xc(0)| 
   &=&   \displaystyle \varepsilon \left|\int_0^t 
           \gcss(x(s))\xss(s)\xu(s) + 
           \sum_{\ell=1}^N \gcs{\ell}(x(s))\xs{\ell}(s)\xu(s) 
         \diff s\right|
\\ &\le& \displaystyle \varepsilon C \int_0^t 
           \exp(\muu - \tilde\alpha\muss) |\xss(0)\xu(0)| +
           \sum_{\ell=1}^N \exp(-\tilde\alpha\mus{\ell}) |\xs{\ell}(0)\xu(s)| \diff s
\\ &\le& 2\displaystyle \varepsilon C \int_0^t 
           \exp(\muu - \tilde\alpha\muss) |\xu(0)| +
           \sum_{\ell=1}^N \exp(-\tilde\alpha\mus{\ell}) |\xs{\ell}(0)| \diff s
\\ &\le& 2\displaystyle \varepsilon C \int_0^t 
           \exp(-\alpha(\muss-\muu)) |\xu(0)| +
           \sum_{\ell=1}^N \exp(-\alpha\mus{\ell}) |\xs{\ell}(0)| \diff s
\end{array}
\]
The last inequality needs a slight adjustment of $\tilde\alpha\leadsto\alpha$ 
and uses the eigenvalue condition (\ref{eqCondLocEigenvalues}). 
Indeed, for all $0<\alpha<1$, we find a suitable $0<\tilde\alpha<1$ with 
$0 < \alpha(\muss-\muu) < \tilde\alpha\muss - \muu$.
\end{proof}

The local map 
\begin{equation}\label{eqLocalMap}
  (\xuIn, \xsIn{}, \xcIn) \;\longmapsto\; 
  (\xssOut, \xsOut{}, \xcOut) \;=\; 
    \MapLocal (\xuIn, \xsIn{}, \xcIn)
\end{equation}
is given by the first intersection of the solution of 
(\ref{eqSplitLocalSystem}) 
to the initial value $(\xuIn, \xssIn=1, \xsIn{}, \xcIn)$
with the out-section $\{\xu=1\}$. See figure \ref{figLocalMap}.
The local map $\MapLocal$ is well-defined on the in-section
\begin{equation}\label{eqInSection}
\SectionIn \;=\; 
  \{\; (\xuIn, \xssIn, \xsIn{}, \xcIn) \;|\;
  \xssIn = 1,\; 0 < \xuIn < 1,\; \|\xsIn{}\| < 1,\;|\xcIn| < 1 \;\}
\end{equation}
see lemma \ref{thLocalPassage} below.
The singular points in the intersection of the stable manifold 
of the equilibrium line with the in-section
are mapped to the respective points 
in the intersection of the unstable manifold of the equilibrium line 
with the out-section:
\begin{equation}\label{eqLocalMapExtension}
\MapLocal(\xuIn=0, \xsIn{}, \xcIn) \;=\; 
  (\xssOut, \xsOut{}, \xcOut) \,:=\; (0,0,\xcIn).
\end{equation}
Note that there is no drift in $\xc$ at the boundary 
due to the invariant fibres (\ref{eqGeneralInvManifolds}).

\begin{lem}\label{thLocalPassage}
There exists an $\varepsilon_0>0$ such that for all $\varepsilon<\varepsilon_0$ 
and $x(0)=\xIn$ in the in-section $\SectionIn$, see (\ref{eqInSection}),
the trajectory $x(t)$ under the flow to the vector field (\ref{eqSplitLocalSystem})
remains in $\neighborhood$ as long as $|\xu| \le 1$, 
i.e.~all along the passage defining the local map $\MapLocal$. 
The passage time $\tloc$ is given by
\begin{equation}\label{eqLocalTime}
 \tloc \;=\; \frac{1}{\muu} \ln \frac{1}{|\xuIn|}.
\end{equation}
\end{lem}
\begin{proof}
We choose $\varepsilon_0$ smaller than
$\frac{1}{2(N+1)C}\alpha\min\{\muss-\muu, \mus{1}, \ldots, \mus{N}\})$, 
see lemma \ref{thLocalBounds}.
Then trajectories starting in $\SectionIn$ cannot leave $\neighborhood$ 
unless $\xu$ becomes larger than 1, see (\ref{eqLocalStrongStableBound}, 
\ref{eqLocalStableBound}, \ref{eqLocalCenterBound}).
Furthermore, (\ref{eqLocalUnstableBound}) ensures that $\xu$ must grow beyond 1.
Thus every trajectory starting in $\SectionIn$ intersects 
the out-section $\SectionOut=\{x_u=1\}$ before leaving $\neighborhood$.
Setting $\xu(\tloc) = 1$ in (\ref{eqLocalUnstableBound}) 
determines the passage time $\tloc$.
\end{proof}

\begin{cor}\label{thContinuousLocalMap}
The local map $\MapLocal$ (\ref{eqLocalMap}, \ref{eqLocalMapExtension}), 
i.e.~the local passage on the closed in-section $\overline\SectionIn$ 
including the singular boundary $\{\xuIn = 0\}$, 
is continuous.
For all $0<\alpha<1$ there exists an $\varepsilon_0>0$ such that 
for all $\varepsilon<\varepsilon_0$ the following estimates hold
\[\begin{array}{rcl}
|\xcOut - \xcIn| &\le& \varepsilon C 
  \left( |\xuIn| + |\xsIn{1}| + \cdots + |\xsIn{N}| \right), \\
|\xssOut| &\le& |\xuIn|^{\alpha\muss/\muu - 1} |\xuIn|, \\ 
|\xsOut{\ell}| &\le& |\xuIn|^{\alpha\mus{\ell}/\muu} |\xsIn{\ell}|,
\qquad \ell=1,\ldots,N,
\end{array}\] 
with $C$ independent of $\varepsilon$ and $\xIn$.

Thus the drift along the line of equilibria is arbitrarily small and the 
distance from the orbit to the union of the stable and unstable manifolds 
shrinks arbitrarily fast, close to the critical orbit.
\end{cor}
\begin{proof}
The estimates follow directly from lemma \ref{thLocalBounds}
applied to the local passage time given by lemma \ref{thLocalPassage}.
They also establish continuity of the local map $\MapLocal$ 
at the singular boundary $\{ \xuIn = 0 \}$.
Note $0 < \xuIn < 1$ on the in-section. 
Note further that $\alpha\muss/\muu - 1 > 0$ 
for $\alpha$ chosen close enough to $1$.
\end{proof}

Unfortunately, there is no hope to obtain Lipschitz estimates for the
local map $\MapLocal$ with respect to the standard metric. 
Even for $g=0$, the linear vector field (\ref{eqGeneralLocalVectorFieldLinNonlin})
yields a non-Lipschitz local passage, for $\mus{\ell} < \muu$.

To obtain Lipschitz bounds for the local map $\MapLocal$, 
we have to introduce a non-Euclidean metric on the in- and out-sections.
We define the Riemannian metrics
\begin{equation}\label{eqMetricIn}
\diff s^2_* \;=\;
    \frac{\|\xus\|_2^2}{\xu^2}\diff \xu^2
  + \sum_{\ell=1}^N \frac{\|\xus\|_2^2}{\xs{\ell}^2}\diff \xs{\ell}^2
  + \diff \xc^2
\end{equation}
on the in-section $\SectionIn$ and
\begin{equation}\label{eqMetricOut}
\diff s^2_* \;=\;
    \frac{\|\xsss\|_2^2}{\xss^2}\diff \xss^2
  + \sum_{\ell=1}^N \frac{\|\xsss\|_2^2}{\xs{\ell}^2}\diff \xs{\ell}^2
  + \diff \xc^2
\end{equation}
on the out-section $\SectionOut$. We denoted the Euclidean norms
$\|\xus\|_2^2=\|(\xu,\xs{})\|_2^2 = \xu^2 + \sum_{\ell=1}^N \xs{\ell}^2$,
$\|\xsss\|_2^2=\|(\xss,\xs{})\|_2^2 = \xss^2 + \sum_{\ell=1}^N \xs{\ell}^2$.
The distance in $\SectionIn$, $\SectionOut$ is then given by the length of 
the shortest connecting paths and denoted by $\dist_*$.
On fibres $\{\xc\;\textrm{fixed}\}$ we define the metric analogously.

Let us discuss the new metric in the cone $\xus = (\xu,\xs{}) \in [0,\infty)^{N+1}$
in $\SectionIn$, ignoring the $\xc$-direction that remains unchanged.
The metric becomes singular along the invariant boundaries 
$\{\xu=0\}, \{\xs{\ell}=0\}$, see (\ref{eqCondLocSubspaces}). 
Inside the open cone $(0,\infty)^{N+1}$,
the new metric $\diff s_*$ is locally equivalent to the Euclidean metric $\diff s$,
\begin{equation}\label{eqDistanceLocalEquivalence}
  \diff s^2 \;\le\; \diff s^2_* \;\le\; 
  (N+1) \frac{\max\{|\xu|,|\xs{1}|,\ldots,|\xs{N}|\}}
             {\min\{|\xu|,|\xs{1}|,\ldots,|\xs{N}|\}} \diff s^2,
\end{equation}
and thus induces the same topology. The origin can be included.
In fact the distance of any point to the origin is bounded by
\begin{equation}\label{eqDistanceToOrigin}
  \|\xus\|_2 \;\le\; \dist_*( 0, \xus ) \;\le\;
  (N+1)^{3/2} \|\xus\|_2.
\end{equation}
(The upper bound can easily obtained by connecting the origin to $\xus$
with a piecewise linear path along the space diagonals with respect 
to suitable coordinate directions.)
Every curve in the open cone hitting the boundary away from the origin 
has infinite length.

In particular, the new metric is uniformly equivalent to the Euclidean metric
in any closed cone that has finite, nonzero angle to the boundaries, i.e.
\[ 
\{\;(\xus,\xc)\;|\;
\|\xus\|_2 \ge c\max\{|\xu|,|\xs{1}|,\ldots,|\xs{N}|\}\;\},
\qquad c > 1.
\]
Thus Lipschitz estimates with respect to the new metric carry over to the 
Euclidean metric.

We denote the in- and out-sections without the singular boundaries but with the 
origin by
\begin{equation}\label{eqNonsingularSections}
\begin{array}{rcl}
\SectionInInterior 
&:=&  \SectionIn \cap ((0,\infty)^{N+1}\cup\{0\})\times\setR)
\\&=& \{\; (\xus,\xc)\in\SectionIn \;|\; 
            \xus=0 \mbox{ or } \xu\xs{1}\cdots\xs{N} \ne 0 \;\},
\\
\SectionOutInterior 
&:=&  \SectionOut \cap ((0,\infty)^{N+1}\cup\{0\})\times\setR)
\\&=& \{\; (\xsss,\xc)\in\SectionOut \;|\; 
            \xsss=0 \mbox{ or } \xss\xs{1}\cdots\xs{N} \ne 0 \;\}.
\end{array}
\end{equation}

Corollary \ref{thContinuousLocalMap} yields Lipschitz continuity 
of the local passage $\MapLocal$ at the origin 
with respect to the Euclidean metric, 
and by (\ref{eqDistanceToOrigin}) also with respect to the new metrics 
(\ref{eqMetricIn}, \ref{eqMetricOut}).
To obtain Lipschitz estimates away from the invariant boundaries 
with respect to the new metric,
we consider the linearisation of the vector field (\ref{eqSplitLocalSystem}) 
along a trajectory $x(t)$ from the in- to the out-section 
to obtain bounds on $D\MapLocal$.

We start with a tangent vector $\dxIn=(\dxuIn,\dxssIn=0,\dxsIn{},\dxcIn)$ 
of unit length with respect to the metric (\ref{eqMetricIn}) 
at a point $\xIn=(\xusIn,\xcIn)\in\SectionIn$,
\begin{equation}\label{eqNormedTangentVector}
  1 \;=\; \|\dxIn\|_*^2 \;=\; 
  \frac{\|\xusIn\|_2^2}{|\xuIn|^2} |\dxuIn|^2
  + \sum_{\ell=1}^N \frac{\|\xusIn\|_2^2}{|\xsIn{\ell}|^2} 
      |\dxsIn{\ell}|^2
  + |\dxcIn|^2.
\end{equation}
First, we project $\dxIn$ along the vector field $f$ into 
the hyperplane $\{\dxu=0\}$, as this remains invariant under 
the linearised flow and corresponds to the out-section.
The projected vector
\begin{equation}\label{eqProjectedTangentVector}
  \dx(0) \;=\; \dxIn - \frac{\dxuIn}{\muu\xuIn} f(\xIn)
\end{equation}
thus represents our initial condition to the linearised flow
\begin{equation}\label{eqLinearizedSystem}
\begin{array}{rcl}
\dot\dxss &=& \left(-\muss(\xc) + \varepsilon\gss(x) 
                    + \varepsilon(\partial_{\xss}\gss(x))\xss\right) \dxss
\\&&~       + \varepsilon \sum_{\ell=1}^N  
                      (\partial_{\xs{\ell}}\gss(x))\xss \dxs{\ell}
\\&&~       + \left(-\muss'(\xc) 
                    + \varepsilon(\partial_{\xc}\gss(x))\xss\right) \dxc,
\\
\dot\dxs{\ell} &=& \left(-\mus{\ell}(\xc) + \varepsilon\gs{\ell}(x) 
                    + \varepsilon(\partial_{\xs{\ell}}\gs{\ell}(x))\xs{\ell}\right) \dxs{\ell}
\\&&~       + \varepsilon (\partial_{\xss}\gs{\ell}(x))\xs{\ell} \dxss
\\&&~       + \varepsilon \sum_{\tilde\ell\ne\ell}  
                      (\partial_{\xs{\tilde\ell}}\gs{\ell}(x))\xs{\ell} \dxs{\tilde\ell}
\\&&~       + \left(-\mus{\ell}'(\xc) 
                    + \varepsilon(\partial_{\xc}\gs{\ell}(x))\xs{\ell}\right) \dxc,
\\
\dot\dxc &=& \varepsilon \left( \partial_{\xc}\gcss(x)\xss\xu 
                    + \sum_{\ell=1}^N \partial_{\xc}\gcs{\ell}(x)\xs{\ell}\xu \right) \dxc
\\&&~       + \varepsilon \left( \gcss(x)\xu + \partial_{\xss}\gcss(x)\xss\xu 
                    + \sum_{\ell=1}^N \partial_{\xss}\gcs{\ell}(x)\xs{\ell}\xu \right) \dxss
\\&&~       + \varepsilon \sum_{\ell=1}^N \left( \partial_{\xs{\ell}}\gcss(x)\xss\xu 
                    + \gcs{\ell}(x)\xu
                    + \sum_{\tilde\ell=1}^N \partial_{\xs{\ell}}\gcs{\tilde\ell}(x)\xs{\tilde\ell}\xu \right) \dxs{\ell}.
\end{array}
\end{equation}
Here we already dropped the $u$-component.

\begin{lem}\label{thInitialTangentVectorEstimate}
Let a \textbf{unit} tangent vector $\dxIn$ to 
$\xIn=(\xusIn,\xcIn)\in\SectionIn$ with respect to the metric 
(\ref{eqMetricIn}) be given. 
The projection $\dx(0)$ of $\dxIn$ 
along the vector field (\ref{eqSplitLocalSystem})
into the plane $\{\dxu=0\}$ is estimated by
\[
\begin{array}{rcl}
|\dxss(0)|       &\le& C \|\xusIn\|_2^{-1}, \\
|\dxs{\ell}(0)|  &\le& C \|\xusIn\|_2^{-1}|\xsIn{\ell}|, \\
|\dxc(0)-\dxcIn| &\le& \varepsilon C,
\end{array}
\]
with a constant $C$ independent of $\xIn$, $\dxIn$, 
and $\varepsilon < \varepsilon_0$, 
provided $\varepsilon_0$ is chosen small enough.
\end{lem}
\begin{proof}
Apply (\ref{eqProjectedTangentVector}) to (\ref{eqNormedTangentVector}) 
and use the bounds (\ref{eqGeneralHotBounds}) on the nonlinear terms of the 
vector field (\ref{eqSplitLocalSystem}).

Indeed, we find
\[
\begin{array}{rcl}
\dx(0) &=&
  \dxIn - \frac{\dxuIn}{\muu\xuIn} f(\xIn)
\\ &=&
  \left(\begin{array}{l}
     \dxuIn\\\dxssIn=0\\\dxsIn{\ell}\\\dxcIn
  \end{array}\right)
  - \displaystyle\frac{\dxuIn}{\muu\xuIn}
  \left(\begin{array}{l}
     \muu\xuIn\\
     \left(-\muss(\xcIn) + \varepsilon \gss(\xIn)\right)\,1\\
     \left(-\mus{\ell}(\xcIn) + \varepsilon \gs{\ell}(\xIn)\right)\xsIn{\ell}\\
     \varepsilon\left(\gcss(\xIn)\,1 + \sum_{\ell=1}^N(\xIn)\xsIn{\ell}\right)
       \xuIn
  \end{array}\right).
\end{array}
\]
Immediately, we have $\dxuIn(0)=0$.
For the other components we again use 
the uniform bounds (\ref{eqGeneralHotBounds}) on the nonlinearity $g$,
the bounds (\ref{eqEigenvalueBounds}) on the eigenvalues,
and the bounds (\ref{eqNormedTangentVector}) on the components of $\dxIn$.
We obtain for arbitrary $0<\alpha<1$, 
if $\varepsilon < \varepsilon_0$ is chosen small enough:
\[
|\dxss(0)| \;\le\; \frac{|\dxuIn|}{\muu|\xuIn|}\frac{1}{\alpha}\muss
\;\le\; \frac{\muss}{\alpha\muu} \frac{1}{\|\xusIn\|_2},
\]
for the component transverse to the in-section, 
\[
|\dxs{\ell}(0)| 
\;\le\; |\dxsIn{\ell}| 
   + \frac{|\dxuIn|}{\muu|\xuIn|}\frac{1}{\alpha}\mus{\ell}|\xsIn{\ell}|
\;\le\; \left(1+\frac{\mus{\ell}}{\alpha\muu}\right) 
  \frac{|\xsIn{\ell}|}{\|\xusIn\|_2},
\]
for each of the remaining $N$ stable components, and
\[
|\dxc(0)-\dxcIn| 
\;\le\; \frac{|\dxuIn|}{\muu|\xuIn|} \varepsilon C (N+1) |\xuIn|
\;\le\; \varepsilon \frac{1}{\muu} C \frac{|\xuIn|}{\|\xusIn\|_2}
\;\le\; \varepsilon \frac{1}{\muu} C,
\]
for the centre component. 
An obvious choice of a new constant $C$ yields all claimed estimates.
\end{proof}

\begin{lem}\label{thLinearizedVectorFieldEstimates}
For all $0<\alpha<1$ there exists an $\varepsilon_0>0$ such that 
for all $\varepsilon<\varepsilon_0$, 
the linearised flow (\ref{eqLinearizedSystem}) can be estimated:
\[
\begin{array}{rcl}
|\dxss|'      &\le& - \alpha \muss |\dxss| 
    + \varepsilon C |\xss| \left( |\dxc| + \sum_{\ell=1}^N|\dxs{\ell}| \right),
\\
|\dxs{\ell}|' &\le& - \alpha \mus{\ell} |\dxs{\ell}| 
    + \varepsilon C |\xs{\ell}| \left( |\dxc| + |\dxss| 
          + \sum_{\tilde\ell\ne\ell}|\dxs{\tilde\ell}| \right),
\\
|\dot\dxc|    &\le& \varepsilon C \left( |\xss| + \sum_{\ell=1}^N |\xs{\ell}| \right) |\xu| |\dxc|
    + \varepsilon C |\xu| \left( |\dxss| + \sum_{\ell=1}^N|\dxs{\ell}| \right),

\end{array}
\]
Here $C$ is a constant independent of $\xIn$, $\dxIn$, and $\varepsilon < \varepsilon_0$.
\end{lem}
\begin{proof}
Use the bounds (\ref{eqGeneralHotBounds}) on the nonlinear terms of the 
vector field (\ref{eqSplitLocalSystem}) and the bounds 
(\ref{eqEigenvalueGradients}) on the derivatives of the eigenvalues.
Note that $x\in\neighborhood=[-2,2]^{N+3}$.
This immediately yields the claimed estimates.
\end{proof}

\begin{lem}\label{thLinearizedFlowEstimates}
For all $0<\alpha<1$ there exists an $\varepsilon_0>0$ such that 
for all $\varepsilon<\varepsilon_0$ the following statement holds:
Let a trajectory $x(t)$ of local passage, 
$x(0)=\xIn\in\SectionInInterior$, $\xusIn\ne0$ be given.
Let a unit tangent vector $\dxIn$ to $\xIn\in\SectionIn$ 
with respect to the metric (\ref{eqMetricIn}), 
and its projection $\dx(0)$ be given.
Then the evolution of $\dx$ under the linearised flow is estimated by
\[
\begin{array}{rcl}
|\dxss(t)|       &\le& \exp(-\alpha^2\muss t) C \|\xusIn\|_2^{-1}, \\
|\dxs{\ell}(t)|  &\le& \exp(-\alpha^2\mus{\ell} t) C \|\xusIn\|_2^{-1} |\xsIn{\ell}|, \\
|\dxc(t)-\dxcIn| &\le& \varepsilon C,
\end{array}
\]
all along the local passage, $t\in[0,\tloc]$.
The constant $C$ is independent of $\xIn$, $\dxIn$, $t$, 
and $\varepsilon < \varepsilon_0$.
\end{lem}
\begin{proof}
Assume
\begin{equation}\label{eqTemporaryBounds}
|\dxss(\tau)| \;\le\; 2 C \|\xusIn\|_2^{-1}, \qquad
|\dxs{\ell}(\tau)| \;\le\; 2 C, \qquad
|\dxc(\tau)| \;\le\; 2,
\end{equation}
on $\tau\in[0,t]$. This assumption holds for small $t$ due the estimates of 
the initial values in lemma \ref{thInitialTangentVectorEstimate}. 
The constant $C$ is taken from the lemma.

Then the estimates of lemmata \ref{thLinearizedVectorFieldEstimates} 
and \ref{thLocalBounds} yield
\[
\begin{array}{rcl}
|\dxss(\tau)|'      &\le& -\alpha\muss|\dxss(\tau)| 
      + \varepsilon C \exp(-\alpha\muss\tau) \left( 2 + 2 N C \right)
\\                  &\le& -\alpha\muss|\dxss(\tau)| 
      + \varepsilon \tilde{C} \exp(-\alpha\muss\tau),
\\
|\dxs{\ell}(\tau)|' &\le& -\alpha\mus{\ell}|\dxs{\ell}(\tau)| 
      + \varepsilon C \exp(-\alpha\mus{\ell}\tau)|\xsIn{\ell}| 
        \left(2 + 2 C \|\xusIn\|_2^{-1} + 2(N-1)C \right)
\\                  &\le& -\alpha\mus{\ell}|\dxs{\ell}(\tau)| 
      + \varepsilon \tilde{C} \exp(-\alpha\mus{\ell}\tau)
        \|\xusIn\|_2^{-1}|\xsIn{\ell}|.
\end{array}
\]
The last inequality uses $\|\xusIn\|_2 < \sqrt{N}$. 
The new constant is bounded by $\tilde{C} \le (2+2NC)\sqrt{N}C$ 
and will again be denoted by $C$.

Then we can integrate
\footnote{This can also be seen as a Gronwall estimate 
for $\exp(\alpha\muss t)\dxss(t)$ and $\exp(\alpha\mus{\ell} t)\dxs{\ell}(t)$.} 
the above estimates to obtain
\[
\begin{array}{rcl}
|\dxss(t)| &\le& \exp(-\alpha\muss t) ( |\dxss(0)| + \varepsilon C t )
\\
|\dxs{\ell}(t)| &\le& 
    \exp(-\alpha\mus{\ell} t) ( |\dxs{\ell}(0)| 
    + \varepsilon C t \|\xusIn\|_2^{-1}|\xsIn{\ell}|)
\end{array}
\]
For small enough $\varepsilon_0$ this yields
\[
\begin{array}{rcl}
|\dxss(t)|
   &\le& \exp(-\alpha^2\muss t) ( |\dxss(0)| + 1 )
\\ &\le& \exp(-\alpha^2\muss t) ( C \|\xusIn\|_2^{-1} + 1 )
\\ &\le& \exp(-\alpha^2\muss t) ( C + 1 ) \|\xusIn\|_2^{-1},
\\
|\dxs{\ell}(t)| 
   &\le& \exp(-\alpha^2\mus{\ell} t) 
         ( |\dxs{\ell}(0)| + \|\xusIn\|_2^{-1}|\xsIn{\ell}| )
\\ &\le& \exp(-\alpha^2\mus{\ell} t) 
         ( C + 1 ) \|\xusIn\|_2^{-1}|\xsIn{\ell}|
\end{array}
\]
In particular, assumptions (\ref{eqTemporaryBounds}) and the first two claims 
hold as long as $|\dxc(\tau)| \le 2$,
if the original constant was chosen larger than $1$.

We use the new estimates of $\dxss, \dxs{\ell}$, the assumption on $\dxc$, 
and the bound on the trajectory given by lemma \ref{thLocalBounds}
to estimate the centre component:
\[
\begin{array}{rcl}
|\dot\dxc(\tau)|
   &\le& \varepsilon C \left( |\xss| + \sum_{\ell=1}^N |\xs{\ell}| \right) |\xu| |\dxc|
         + \varepsilon C |\xu| \left( |\dxss| + \sum_{\ell=1}^N|\dxs{\ell}| \right)
\\ &\le& 2 \varepsilon C \left( |\xss| + \sum_{\ell=1}^N |\xs{\ell}| \right) |\xu|
\\ &&       + \varepsilon C (C+1) \|\xusIn\|_2^{-1} |\xu| 
              \left( \exp(-\alpha^2\muss \tau) 
              + \sum_{\ell=1}^N \exp(-\alpha^2\mus{\ell} \tau)|\xsIn{\ell}| \right)
\\ &\le& 2 \varepsilon C 
              \left( \exp(-(\alpha\muss-\muu)\tau) |\xuIn| 
              + \sum_{\ell=1}^N \exp(-\alpha\mus{\ell}\tau)|\xsIn{\ell}| \right)
\\ &&       + \varepsilon C (C+1) \|\xusIn\|_2^{-1} 
              \left( \exp(-(\alpha^2\muss-\muu) \tau)|\xuIn| 
              + \sum_{\ell=1}^N \exp(-\alpha^2\mus{\ell} \tau)|\xsIn{\ell}| \right)
\\ &\le& \varepsilon\tilde{C} \left( \exp(-(\alpha^2\muss-\muu) \tau) 
                     + \sum_{\ell=1}^N \exp(-\alpha^2\mus{\ell} \tau) \right),
\end{array}
\]
with new constant $\tilde{C} < C(C+3)$. 
Thus $|\dot\dxc(\tau)|$ decays exponentially for $\alpha$ close enough to 1. 
Integration yields
\[
\int_0^t |\dot\dxc(\tau)| \;\le\; \varepsilon \hat{C},
\]
with 
$\hat{C} < \tilde{C} ( (\alpha^2\muss-\muu)^{-1} 
              + \sum_{\ell=1}^N (\alpha^2\mus{\ell})^{-1} )$.
If $\varepsilon_0$ is chosen small enough, 
this shows that the assumption (\ref{eqTemporaryBounds})
indeed holds all along the passage and the claimed estimates are valid.
\end{proof}

\begin{thm}[local Lipschitz map]\label{thLocalMapLipschitz}
The local passage 
$\MapLocal:\SectionInInterior\to\SectionOutInterior$
is Lipschitz continuous with respect to the metrics 
(\ref{eqMetricIn}, \ref{eqMetricOut}).

There exist $\beta>0$, $\varepsilon_0>0$ and $C>0$ such that 
for all $\varepsilon<\varepsilon_0$ the following estimates hold for all
$\xIn, \tilde\xIn$ with $0 \le \tilde\xuIn \le \xuIn$:
\[
\begin{array}{rcl}
\dist_*(\tilde\xsssOut, \xsssOut)
  &\le& |\xuIn|^\beta C \;\dist_*(\tilde\xIn,\xIn),
\\
|(\tilde\xcOut-\xcOut) - (\tilde\xcIn-\xcIn)| 
  &\le& \varepsilon C \;\dist_*(\tilde\xIn,\xIn).
\end{array}
\]

The domain $\SectionInInterior$, as defined in (\ref{eqNonsingularSections}),
is given by the local section without the invariant singular boundaries
but including the line $(0,\xc)$ representing the cap of heteroclinic orbits.

The drift in the centre direction can be made arbitrarily small by choosing 
a sufficiently small local neighbourhood.
The contraction in the transverse directions is arbitrarily strong by 
restricting the in-section to the part close to the primary object, 
i.e.~the stable manifold of the origin.
\end{thm}
\begin{proof}
This is a corollary of lemma \ref{thLinearizedFlowEstimates} by applying 
the passage time (\ref{eqLocalTime}) and the metric (\ref{eqMetricOut}). 
Extension to the line $(0,\xc)$ is given by corollary \ref{thContinuousLocalMap}.

Indeed, in the out-section, the estimates of lemma \ref{thLinearizedFlowEstimates},
read
\[
\begin{array}{rclcl}
|\dxssOut|      
  &=& |\dxss(\tloc)|       
  &\le& |\xuIn|^{\alpha^2\muss/\muu} C \|\xusIn\|_2^{-1}, \\
|\dxsOut{\ell}|  
  &=& |\dxs{\ell}(\tloc)| 
  &\le& |\xuIn|^{\alpha^2\mus{\ell}/\muu} C \|\xusIn\|_2^{-1} |\xsIn{\ell}|, \\
|\dxcOut-\dxcIn| 
  &=& |\dxc(\tloc)-\dxcIn| 
  &\le& \varepsilon C.
\end{array}
\]
With respect to the modified metric (\ref{eqMetricOut}) we find 
using the estimates of lemma \ref{thLocalBounds}:
\[
\begin{array}{rcl}
\displaystyle \left(\frac{\|\xsssOut\|_2}{|\xssOut|}|\dxssOut|\right)^2
&\le& \displaystyle
        \left( 1 + \sum_{\ell=1}^N 
                   |\xuIn|^{2(\alpha\mus{\ell}/\muu - \muss/\alpha\muu)}
                   |\xsIn{\ell}|^2 
        \right) |\dxssOut|^2 
\\&\le& \displaystyle C
        \frac{\displaystyle |\xuIn|^{2\xi} 
              + \sum_{\ell=1}^N 
                |\xuIn|^{2(\xi + \alpha\mus{\ell}/\muu - \muss/\alpha\muu)}
                |\xsIn{\ell}|^2}
             {\displaystyle |\xuIn|^2
              + \sum_{\ell=1}^N |\xsIn{\ell}|^2}
             |\xuIn|^{2(\alpha^2\muss/\muu-\xi)}
\\&\le& |\xuIn|^{2\beta} C.
\end{array}
\]
In the second inequality, we introduced a parameter $\xi$. 
The last inequality then needs
\[
\begin{array}{rclcl}
0 &<& \alpha^2\muss/\muu-\xi &=:& \beta,
\\
1 &\le& \xi,
\\
0 &\le& \xi + \alpha\mus{\ell}/\muu - \muss/\alpha\muu,
&& \ell=1,\ldots,N.
\end{array}
\]
For $\alpha$ close to $1$, a suitable $\xi$ exists. 
In fact we can obtain arbitrary
\[
  0 \;<\; \beta \;<\; \min \{\; \muss/\muu-1, \; \mus{\ell}/\muu \;\}.
\]
Similarly we find for $1\le\ell\le N$,
\[
\begin{array}{rcl}
&& \hspace{-3em}
\displaystyle \left(\frac{\|\xsssOut\|_2}{|\xsOut{\ell}|}|\dxsOut{\ell}|\right)^2
\\[2ex]&\le& \displaystyle
        \left( |\xuIn|^{2(\alpha\muss/\muu - \mus{\ell}/\alpha\muu)} 
               + \sum_{\tilde\ell=1}^N 
                 |\xuIn|^{2(\alpha\mus{\tilde\ell}/\muu - \mus{\ell}/\alpha\muu)}
                 |\xsIn{\tilde\ell}|^2 \right) 
        \frac{|\dxsOut{\ell}|^2}{|\xsIn{\ell}|^2} 
\\&\le& \displaystyle C
        \frac{\displaystyle |\xuIn|^{2(\xi + \alpha\muss/\muu - \mus{\ell}/\alpha\muu)} 
              + \sum_{\tilde\ell=1}^N 
                |\xuIn|^{2(\xi + \alpha\mus{\tilde\ell}/\muu - \mus{\ell}/\alpha\muu)}
                |\xsIn{\tilde\ell}|^2}
             {\displaystyle |\xuIn|^2
              + \sum_{\tilde\ell=1}^N |\xsIn{\tilde\ell}|^2}
             |\xuIn|^{2(\alpha^2\mus{\ell}/\muu-\xi)}
\\&\le& |\xuIn|^{2\tilde\beta} C.
\end{array}
\]
This time we need for the last inequality
\[
\begin{array}{rclcl}
0 &<& \alpha^2\mus{\ell}/\muu-\xi &=:& \tilde\beta,
\\
1 &\le& \xi + \alpha\muss/\muu - \mus{\ell}/\alpha\muu,
\\
0 &\le& \xi + \alpha\mus{\tilde\ell}/\muu - \mus{\ell}/\alpha\muu,
&&\tilde{\ell}=1,\ldots,N.
\end{array}
\]
Again, for $\alpha$ close to $1$, a suitable $\xi$ exist.
In fact, we can again obtain arbitrary
\[
  0 \;<\; \tilde\beta \;<\; \min \{\; \muss/\muu-1, \; \mus{\ell}/\muu \;\}.
\]

Now, take a geodesic curve in $\SectionIn$ that defines $\dist_*(\tilde\xIn,\xIn)$.
The image of this curve under the passage $\MapLocal$ provides an upper bound 
on $\dist_*(\tilde\xOut,\xOut)$. 
In both sections the $\xc$-component can be separated.
Therefore the above estimates on the evolution of the tangent vectors 
immediately yield the claims of the theorem.
\end{proof}

\begin{rem}\label{thUniformLocalEstimates}
In theorem \ref{thLocalMapLipschitz}, the constant $C$ only depends on the 
$\mathcal{C}^1$ bounds on the nonlinear part of the vector field 
and the derivatives of the eigenvalues of the linearisation 
along the line of equilibria.
The exponent $\beta$ only depends on the spectral gaps. 
In fact, it can be taken arbitrarily in the interval
\[
  0 \;<\; \beta \;<\; \min \{\; \muss/\muu-1, \; \mus{\ell}/\muu \;\},
\]
by choosing $\varepsilon_0$ small enough.
\end{rem}

The last remark provides uniform Lipschitz estimates for the local passages
near the Kasner circle in Bianchi models, provided they keep a uniform distance 
from the Taub points at which the spectral gap shrinks to zero.


\section{Return map and graph transform}
\label{secReturnMap}

In this section we define a global excursion map for trajectories near 
a primary heteroclinic orbit to the Kasner circle, 
that is from the out section of a local passage to the in section of 
another local passage,
both local passages as discussed in the previous section.
Combining local passage and global excursion we obtain a return map from 
one in-section to the next,
\begin{equation}\label{eqReturnMap}
\MapCombined_k \,:=\; \MapGlobal_k \circ \MapLocal_k \,:\; 
\SectionIn_k \longrightarrow \SectionIn_{k+1},
\end{equation}
see figure \ref{figCombinedMap}.
The given heteroclinic orbit corresponds to a fixed origin of this map.

\begin{figure}
\centering
\setlength{\unitlength}{0.008\textwidth}
\begin{picture}(120,48)(-60,-5)
\put(-60,-5){\makebox(0,0)[bl]{\includegraphics[width=120\unitlength]{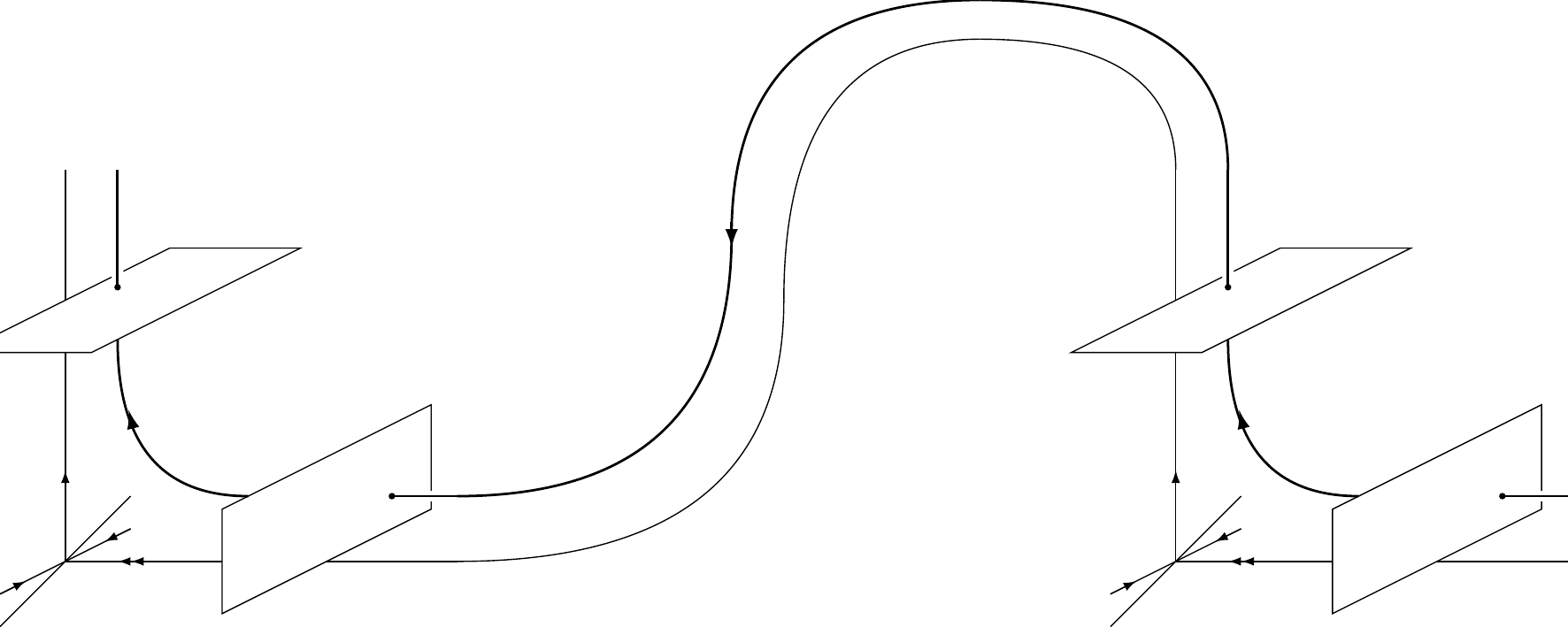}}}
\put(30,0){
\put(0,-1){\makebox(0,0)[tl]{$\,p_{k}$}}
\put(28,8){\makebox(0,0)[l]{$\,\SectionIn_{k}$}}
\put(13,24.5){\makebox(0,0)[b]{$\SectionOut_{k}$}}
\put(7,9){\makebox(0,0)[bl]{$\MapLocal_k$}}
\put(25,5){\makebox(0,0)[r]{\small$\xIn_k$}}
\put(4,21){\makebox(0,0)[l]{\small$\xOut_k$}}
}
\put(-55,0){
\put(0,-1){\makebox(0,0)[tl]{$\,p_{k+1}$}}
\put(28,8){\makebox(0,0)[l]{$\,\SectionIn_{k+1}$}}
\put(13,24.5){\makebox(0,0)[b]{$\SectionOut_{k+1}$}}
\put(7,9){\makebox(0,0)[bl]{$\MapLocal_{k+1}$}}
\put(25,5){\makebox(0,0)[r]{\small$\xIn_{k+1}$}}
\put(4,21){\makebox(0,0)[l]{\small$\xOut_{k+1}$}}
}
\put(0,0){\makebox(0,20)[tl]{$\,h_{k}(\cdot)$}}
\put(-5,25){\makebox(0,0)[r]{$\MapGlobal_k$}}
\end{picture}
\caption{\label{figCombinedMap}
The return map $\MapCombined_k \,:=\; \MapGlobal_k \circ \MapLocal_k$.}
\end{figure}

We prove uniform Lipschitz- and cone properties of the return map,
independently of the given heteroclinic orbit, 
as long as the orbit keeps a uniform distance from the Taub points.
In fact, we prove uniform Lipschitz- and cone properties of the return map
under the conditions \ref{condGlobalExcursion} on the global excursion.

This yields a sequence of return maps, with uniform estimates, 
to every sequence of heteroclinic orbits to the Kasner circle that does not 
accumulate to any Taub point and satisfies the 
local conditions \ref{condLocalPassage} at every equilibrium.

Due to their cone properties, the return maps induce a contracting map 
on a suitable space of sequences of Lipschitz curves.
The fixed point provided by the contraction mapping theorem then yields
the stable manifold of the heteroclinic sequence as claimed in theorems 
\ref{thPeriod3Manifold}, \ref{thAbstractSequenceManifold}.

Take a sequence $p^k$, $k\in\setN$ of equilibria on the Kasner circle,
not accumulating at any Taub point and connected by heteroclinic orbits $h^k(t)$,
$\lim_{t\to\infty}h^{k-1}(t)=p^k=\lim_{t\to-\infty}h^k(t)$, 
as in theorem \ref{thAbstractSequenceManifold}.
Assume that the local conditions \ref{condLocalPassage} hold uniformly at all $p^k$,
in particular $\sup_{k\in\setN}\muu(p^k)/\muss(p^k) < 1$ and 
$\inf_{k\in\setN}\mus{\ell}(p^k) > 0$.
In the Bianchi VI${}_0$ system (\ref{eqBianchiIIMagnetic}) with magnetic field, 
these conditions are satisfied exactly for a chain of heteroclinic orbits 
not accumulating at Taub points and not containing heteroclinic orbits of 
the magnetic family to equilibria in the intervals $\mathcal{K}_2$, $\mathcal{K}_5$, 
see figure \ref{fig3cycle}. 
In particular, the conditions hold for the period 3 cycle.

The previous section then applies to all $p^k$ and the coefficients
$\varepsilon_0, \beta, C$ of the local estimates of 
theorem \ref{thLocalMapLipschitz} can be taken independent of $k$, 
see remark \ref{thUniformLocalEstimates}.

Note the order of fixing the rescaling parameters:
First $\varepsilon_0$ resp.~$\varepsilon$ is fixed small enough to yield our 
estimates of the 
local passages $\MapLocal_k$ with small Lipschitz constants, 
in particular $\varepsilon C \ll 1$ in theorem \ref{thLocalMapLipschitz}. 
This amounts to a choice of the sections $\SectionIn(p_k)$ and $\SectionOut(p_k)$ 
in the original (unscaled) coordinates and also fixes 
the global excursion maps $\MapGlobal_k$.

Due to the non-Euclidean metric used in theorem \ref{thLocalMapLipschitz},
we have to restrict our local passage map to
$\MapLocal_k:\SectionIn_k\to\SectionOut_k$ by 
\begin{equation}\label{eqModifiedSections}
\begin{array}{rclcl}
\SectionIn_k 
&=& \SectionInInterior(p_k)
&=& \{\; (\xus,\xc)\in\SectionIn(p_k) \;|\; 
            \xus=0 \mbox{ or } \xu\xs{1}\cdots\xs{N} \ne 0 \;\},
\\
\SectionOut_k 
&=& \SectionOutInterior(p_k)
&=& \{\; (\xsss,\xc)\in\SectionOut(p_k) \;|\; 
            \xsss=0 \mbox{ or } \xss\xs{1}\cdots\xs{N} \ne 0 \;\},
\end{array}
\end{equation}
see (\ref{eqNonsingularSections}).
Then a sufficiently small upper bound for $\xuIn$ is chosen, 
i.e.~$\MapLocal_k$ are restricted to smaller sections
\begin{equation}\label{eqReducedInSection}
\tilde\SectionIn_k \;=\; 
\{\; x\in\SectionIn_k \;|\; 
     0\le\xu\le\delta,\; 0\le\xs{\ell}\le\delta,\; |\xc| \le \delta \;\} 
\;=\; ((0,\delta]^{N+1} \cup {0}) \times [-\delta,\delta].
\end{equation}
This makes the contraction of the local passage as strong as we like without 
changing $\MapLocal_k$, $\MapGlobal_k$.
It also ensures that trajectories of interest stay close to the Kasner caps 
of heteroclinic orbits and 
therefore the global excursions $\MapGlobal_k$ on the domain of interest are 
as close to the Kasner map as we like.
It also ensures that all non-singular trajectories in these domains indeed 
return to the following in-sections $\SectionIn_{k+1}$.

The global conditions \ref{condGlobalExcursion} hold accordingly: 
the invariant subspaces, (glob-i), are those of the Bianchi system;
the uniform bound, (glob-ii), is fixed by choice of a uniform size $\varepsilon_0$ 
of all local neighbourhoods; and expansion, (glob-iii), is given by the Kasner map.
Uniform expansion again needs a uniform distance from the Taub points.

The following lemma relates the global excursions to the new metric 
used for the local estimates.

\begin{lem}\label{thGlobalBoundsInNewMetric}
Let the global conditions \ref{condGlobalExcursion} be satisfied for the sequence
$\MapGlobal_k:\SectionOut_k\to\SectionIn_{k+1}$ of global excursions.
Then condition $(glob-ii)$ also holds with respect to the new metrics 
(\ref{eqMetricIn}, \ref{eqMetricOut}).
\end{lem}

\begin{proof}
Due to the invariant subspaces, (glob-iii), 
the linearisation $D\MapGlobal_k$ are diagonal at the origin and close 
to diagonal in the neighbourhoods of interest. 
Moreover, the transverse components of $\MapGlobal_k$ have the form
\[
[\MapGlobal_k(x)]_{*} = [\tilde\MapGlobal_k(x)]_{*}x_{*},
\qquad * = \mathrm{ss}, \mathrm{s1}, \ldots, \mathrm{sN},
\]
with smooth $\tilde\MapGlobal_k$.
Now note the definitions (\ref{eqMetricIn}, \ref{eqMetricOut}) of the new metric.
The bounds (glob-ii), on first and second derivatives of $\MapGlobal_k$,
yield uniform bounds on $\tilde\MapGlobal_k$ and their first derivatives.
Thus, the ratio of the coefficients of the metric at an arbitrary
$x\in\SectionOut_k$ to the coefficients at $\MapGlobal_k(x)\in\SectionIn_{k+1}$ 
is between $\|(\tilde\MapGlobal_k)^{-1}\|^{-1}\|\tilde\MapGlobal_k\|^{-1}$ and 
$\|(\tilde\MapGlobal_k)^{-1}\|\|\tilde\MapGlobal_k\|$, 
that is between $M^{-2}$ and $M^2$, for a uniform constant $M$.
This immediately yields new uniform bounds on the derivatives of $\MapGlobal$ 
with respect to the new metric.
\end{proof}

Now we can proceed along the lines of 
\cite{LiebscherHaerterichWebsterGeorgi2011-BianchiA} (section 4)
to establish the existence of stable manifold by a graph-transform approach.

\begin{lem}\label{thLipschitzReturnMap}
Assume conditions \ref{condLocalPassage} on the local passages and 
conditions \ref{condGlobalExcursion} on the global excursions.

Then the return maps (\ref{eqReturnMap}) are Lipschitz continuous with respect 
to the metric (\ref{eqMetricIn}). 
Furthermore, there exist $\varepsilon > 0$, $\delta>0$, $0<\sigma<1$, 
$K_{u,s} > 1$, and $K_\subC > (1-\sigma^2)^{-1} > 1$,
such that the following cone conditions hold for
\[
\MapCombined_k \;=\; \MapGlobal_k \circ \MapLocal_k 
  \;:\; \tilde\SectionIn_k \to \SectionIn_{k+1}.
\]

Here $\SectionIn_k$ are the in-sections (\ref{eqModifiedSections}) 
corresponding to the choice of $\varepsilon$, 
and $\tilde\SectionIn_k$ are suitable subsets of the form 
(\ref{eqReducedInSection}).

The cones are defined for $x\in\tilde\SectionIn$ (omitting the index $k$) as
\begin{equation}\label{eqCones}
\begin{array}{rcl}
C_x^\subC &=& \{\tilde{x}\in\tilde\SectionIn\;|\; 
    \dist_*(\tilde\xus, \xus) \le \sigma |\tilde\xc-\xc| \},
\\
C_x^\subUS &=& \{\tilde{x}\in\tilde\SectionIn\;|\; 
    |\tilde\xc-\xc| \le \sigma \dist_*(\tilde\xus-\xus) \}.
\end{array}
\end{equation}
The cone conditions are
\begin{enumerate}
\item[(i)] 
Invariance: 
$\MapCombined(C_x^\subC) \cap \tilde\SectionIn 
  \subset (\mathrm{int\,} C_{\MapCombined x}^\subC) \cup \{\MapCombined x\}$ 
and 
$\MapCombined^{-1}(C_{\MapCombined x}^\subUS) \cap \tilde\SectionIn 
  \subset (\mathrm{int\,} C_x^\subUS) \cup \{x\}$;
\item[(ii)] 
Contraction \& Expansion: 
For all $\tilde{x} \in C_x^\subC$ 
we have expansion in the centre direction: 
$|(\MapCombined\tilde{x})_c-(\MapCombined x)_c| \ge K_\subC |\tilde{x}_c-x_c|$
and for all $\MapCombined\tilde{x}\in C_{\MapCombined x}^\subUS$ 
we have contraction in the transverse directions:
$\dist_*(\tilde{x}_{u,s}, x_{u,s}) \ge 
  K_{u,s} \dist_*((\MapCombined\tilde{x})_{u,s}, (\MapCombined x)_{u,s})$.
\end{enumerate}
They hold for all, $x,\tilde{x},\MapCombined x,\MapCombined\tilde{x} \in \tilde\SectionIn$.
See also figure \ref{figConeConditions}.

The coefficients $\sigma, \delta$ only depend on $\varepsilon_0$ and the 
uniform expansion (glob-iii), that is the distance to the Taub point in the Bianchi system.
\end{lem}

\begin{figure}
\setlength{\unitlength}{0.009\textwidth}
\centering
\begin{picture}(100,40)(-50,-5)
\put(-50,-5){\makebox(0,0)[bl]{\includegraphics[width=100\unitlength]{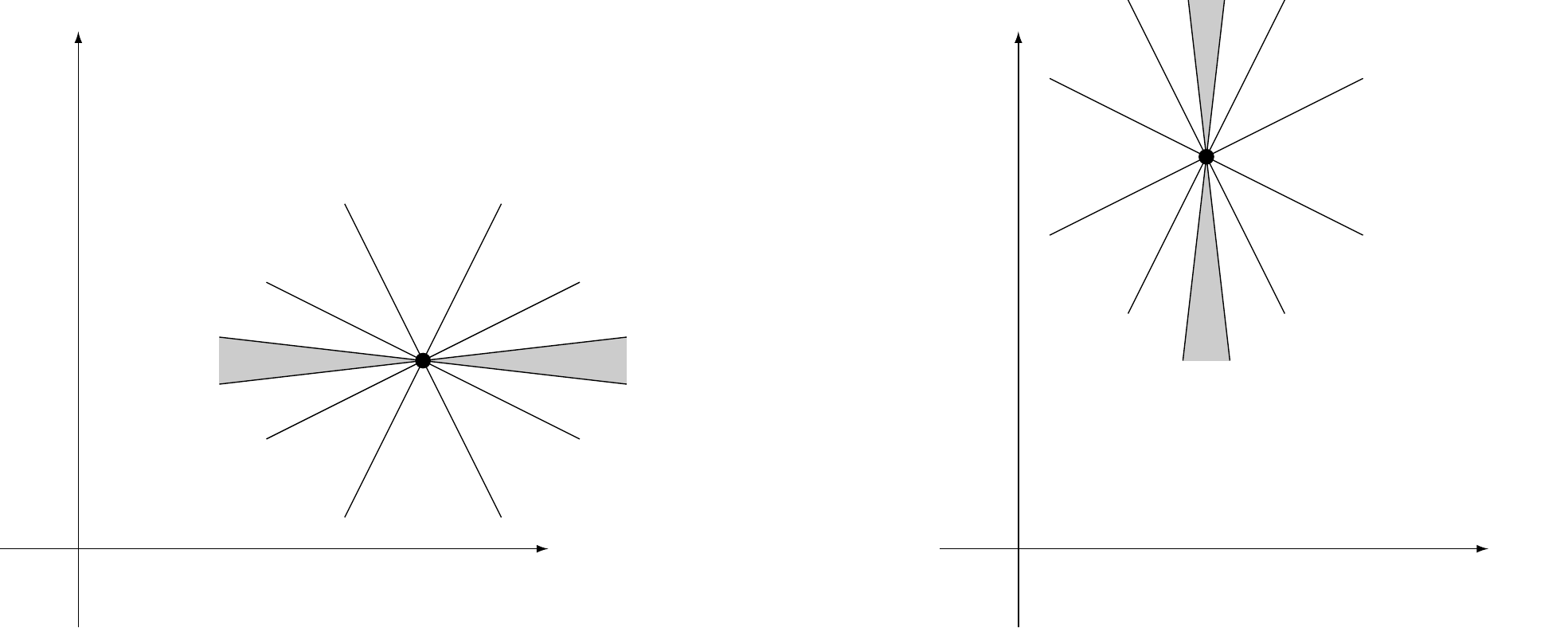}}}
\put(-45,  0){%
\put(0,33){\makebox(0,0)[b]{$x_c$}}
\put(30,0){\makebox(0,0)[l]{$x_{u,s}$}}
\put( 22, 12){%
\put(  0, -2){\makebox(0,0)[t]{$x$}}
\put(-4.2, 8){\makebox(0,0)[bl]{$C^\subC_x$}}
\put(  8,4.2){\makebox(0,0)[tl]{$C^\subUS_x$}}
\put(-13,  0){\makebox(0,0)[r]{$\MapCombined^{-1}(C^\subUS_{\MapCombined x})$}}
}
}
\put( 15,  0){%
\put(0,33){\makebox(0,0)[b]{$\MapCombined(x)_\subC$}}
\put(30,0){\makebox(0,0)[l]{$\MapCombined(x)_\subUS$}}
\put( 12, 25){%
\put( -2,  0){\makebox(0,0)[r]{$\MapCombined(x)$}}
\put(-4.2, 8){\makebox(0,0)[bl]{$C^\subC_{\MapCombined x}$}}
\put(  8,4.2){\makebox(0,0)[tl]{$C^\subUS_{\MapCombined x}$}}
\put(  0,-13){\makebox(0,0)[t]{$\MapCombined(C^\subC_x)$}}
}
}
\put(0,0){\makebox(0,0)[b]{\begin{minipage}{16\unitlength}\centering
{\Huge$\stackrel{\mathbf{\MapCombined}}{\longrightarrow}$}\\
Expansion of $x_c$
Contraction of $x_{u,s}$
\end{minipage}}}
\end{picture}
\caption{\label{figConeConditions}Cone properties of the return map $\MapCombined$.}
\end{figure}

\begin{proof}
Lipschitz continuity of the return map $\MapCombined_k$ follows directly from 
Lipschitz continuity of the local passage $\MapLocal_k$,
see theorem \ref{thLocalMapLipschitz}, as the global excursion $\MapGlobal_k$ is smooth.
To simplify notation, we drop the index $k$ from now on. 
All estimates will be uniform in $k$.

The cone conditions require the expansion in $\xc$-direction given by (glob-iii),
corresponding to the expansion Kasner circle induced by the Kasner map 
of the Bianchi system.
In fact, (glob-iii) states that we have
\begin{equation}\label{eqGeneralKasnerMap}
 \MapGlobal (\xssOut = 0, \xsOut{} = 0, \xcOut) 
  \;=\; (\xuIn = 0, \xsIn{} = 0, \Phi(\xcOut)),
\end{equation}
with Lipschitz continuous $\Phi, \Phi^{-1}$. 
The Lipschitz constant of $\Phi^{-1}$ is less than $L<1$ independent of $k$.
Note again the invariant boundaries, (loc-v), (glob-i).
Therefore we can write, as in the proof of lemma \ref{thGlobalBoundsInNewMetric},
\[
\MapGlobal (\xss, \xs{}, \xc) \;=\; (0,0,\Phi(\xc)) 
    + \tilde\MapGlobal(\xss, \xs{}, \xc) \xsss,
\]
with a smooth matrix $\tilde\MapGlobal(\xss, \xs{}, \xc)$ and vector $\xsss=(\xss,\xs{})$.

Consider now two points $\tilde{x}, x \in\SectionOut$. 
Choose geodesic paths 
$\gamma_1$ from 0 to $\xsss$
and $\gamma_2$ from $\xsss$ to $\tilde\xsss$,
both with respect to the new metric (\ref{eqMetricOut}).
Then 
\[
\begin{array}{rcl}
&&\hspace{-4em}
\dist_*(\MapGlobal(\tilde{x}),\MapGlobal(x))
\\ &\le& 
|\Phi(\tilde\xc) - \Phi(\xc)|
\\&&\displaystyle
+ \int_0^{\dist_*(0,\xsss)} \frac{\diff}{\diff s}
  \dist_*(\MapGlobal(\gamma_1(s),\tilde\xc), \MapGlobal(\gamma_1(s),\xc)) \diff s
\\&&\displaystyle
+ \int_0^{\dist_*(\tilde\xsss,\xsss)} \frac{\diff}{\diff s}
  \dist_*(\MapGlobal(\gamma_2(s),\tilde\xc), \MapGlobal(x)) \diff s
\\ &\le&
|\Phi(\tilde\xc) - \Phi(\xc)|
\\&&\displaystyle
+ \int_0^{\dist_*(0,\xsss)} \frac{\diff}{\diff s}
  \int_{\xc}^{\tilde\xc} \frac{\diff}{\diff t}
  \dist_*(\MapGlobal(\gamma_1(s),t), \MapGlobal(\gamma_1(s),\xc)) \diff t \diff s
\\&&\displaystyle
+ \int_0^{\dist_*(\tilde\xsss,\xsss)} \frac{\diff}{\diff s}
  \dist_*(\MapGlobal(\gamma_2(s),\tilde\xc), \MapGlobal(x)) \diff s,
\end{array}
\]
and we obtain the following Lipschitz estimate
\[
\begin{array}{rcl}
&&\hspace{-4em}
|\;\dist_*(\MapGlobal(\tilde{x}),\MapGlobal(x)) - |\Phi(\tilde\xc) - \Phi(\xc)|\;|
\\ &\le& 
\tilde{C}^\supGlobal\left( \dist_*(0,\xsss)|\tilde\xc-\xc| + \dist_*(\tilde\xsss,\xsss) \right),
\\ &\le& 
C^\supGlobal\left( \|\xsss\||\tilde\xc-\xc| + \dist_*(\tilde\xsss,\xsss) \right),
\end{array}
\]
with $C^\supGlobal$ only depending on the uniform bounds on 
$\|D\MapGlobal\|, \|D^2\MapGlobal\|$ with respect to the new metric 
provided by lemma \ref{thGlobalBoundsInNewMetric}.
The last inequality used the trivial upper bound (\ref{eqDistanceToOrigin})
on the distance from the origin in the new metric.

Using $\MapLocal(\tilde{x})$ and $\MapLocal(x)$ instead of $\tilde{x}$ and $x$
we get a similar estimate for the return map $\MapCombined(x)=\MapGlobal(\MapLocal(x))$:
\begin{equation}\label{eqEstimateReturnMap}
\begin{array}{rcl}
&&\hspace{-4em}
|\;\dist_*(\MapCombined(\tilde{x}),\MapCombined(x)) 
- |\Phi(\MapLocal(\tilde{x})_\subC) - \Phi(\MapLocal(x)_\subC)|\;|
\\ &=& 
|\;\dist_*(\MapGlobal(\MapLocal(\tilde{x})) - \MapGlobal(\MapLocal(x)))
- |\Phi(\MapLocal(\tilde{x})_\subC) - \Phi(\MapLocal(x)_\subC)|\;|
\\ &\le& 
C^\supGlobal\left( 
  \|\MapLocal(x)_\subSSS\||\MapLocal(\tilde{x})_\subC-\MapLocal(x)_\subC| 
  + \dist_*(\MapLocal(\tilde{x})_\subSSS,\MapLocal(x)_\subSSS) \right)
\\ &\le& 
C^\supGlobal\left( |\xu|^\beta \|x\| (1+\varepsilon C) \dist_*(\tilde{x},x)
  + |\xu|^\beta C \dist_*(\tilde{x},x) \right)
\\ &\le& 
C^\supReturn |\xu|^\beta \dist_*(\tilde{x},x).
\end{array}
\end{equation}
The second last inequality uses the estimates of the local passage of 
corollary \ref{thContinuousLocalMap} and theorem \ref{thLocalMapLipschitz} 
for the choice (w.l.o.g.) $0 \le \tilde\xu \le \xu$. 
Note that the estimates of theorem \ref{thLocalMapLipschitz} are used in the form
\[
\begin{array}{rcl}
| \MapLocal(\tilde{x})_\subC - \MapLocal(x)_\subC| 
&\leq& \varepsilon C \dist_*(\tilde{x},x) + |\tilde\xc-\xc| 
\;\leq\; (\varepsilon C+1) \dist_*(\tilde{x},x), 
\\
\dist_*( \MapLocal(\tilde{x})_\subSSS , \MapLocal(x)_\subSSS ) 
&\leq& |\xu|^\beta C \dist_*(\tilde{x},x).
\end{array}
\]
The constant $C^\supReturn$ is uniform in $x$, $\tilde{x}$ in the in-section, 
and the omitted number $k$ of the section along the heteroclinic chain. 
Because $0 < \beta < \min \{\muss/\muu-1, \mus{\ell}/\muu\}$, 
we have an arbitrarily strong contraction  
for $\xu < \delta$, if we choose $\delta$ small enough.

The map $\Phi$ given by (\ref{eqGeneralKasnerMap}),
i.e. the Kasner map in the original Bianchi system,
is expanding, see condition (glob-iii):
\[
  | \Phi(a) - \Phi(b) |    \;\ge\;   L^{-1} |a-b|,
\]
for some uniform constant $L < 1$.

Now choose $K_\subC$ with $1 < K_\subC < L^{-1}$, and $\sigma$ with $0 < \sigma < 1$ 
such that $K_\subC(1-\sigma^2) > 1$.
(The last relation is needed to obtain a contraction in theorem \ref{thStableSet}.)

Consider the cone in centre direction with opening $\vartheta>0$, 
i.e.~$\dist_*(\tilde\xus,\xus) \le \vartheta |\tilde\xc-\xc|$.
Then (\ref{eqEstimateReturnMap}) using the local Lipschitz estimate of 
theorem \ref{thLocalMapLipschitz} yields
\begin{equation}\label{eqEstimateExpansion}
\begin{array}{rcl}
|(\MapCombined(\tilde{x}) - \MapCombined(x))_c| 
&\ge& 
| \Phi(\MapLocal(\tilde{x})_\subC) - \Phi(\MapLocal(x)_\subC) | 
  - C^\supReturn |\xu|^\beta \dist_*(\tilde{x},x)
\\ &\ge& 
L^{-1} | \MapLocal(\tilde{x})_\subC - \MapLocal(x)_\subC | 
  - C^\supReturn |\xu|^\beta \dist_*(\tilde{x},x)
\\ &\ge& 
L^{-1} | \tilde\xc - \xc | 
  - L^{-1} \varepsilon C \dist_*(\tilde{x},x)
  - C^\supReturn |\xu|^\beta \dist_*(\tilde{x},x)
\\ &\ge& 
\left( L^{-1} - \left( L^{-1} \varepsilon C + 
                       C^\supReturn |\xu|^\beta \right)(1+\vartheta) \right) 
   | \tilde\xc - \xc |.
\end{array}
\end{equation}
For $\varepsilon$ and $\delta$ chosen small enough, using $|\xu| \le \delta$, 
we can achieve
\[
 K_\subC \; < \; 
L^{-1} - \left( L^{-1} \varepsilon C + C^\supReturn |\xu|^\beta \right)(1+1/\sigma),
\]
yielding the expansion not only in the cone $C_x^\subC$, with $\vartheta=\sigma<1$, 
but also outside the cone $C_x^\subUS$, with $\vartheta=1/\sigma$.

Furthermore, using again (\ref{eqEstimateReturnMap}), 
we see the invariance of the cones. 
Indeed, assume again $\dist_*(\tilde\xus,\xus) \le \vartheta |\tilde\xc-\xc|$, 
then we have
\[
\begin{array}{rcl}
\dist_*(\MapCombined(\tilde{x})_\subUS , \MapCombined(x)_\subUS ) 
   &\le& C^\supReturn |x_u|^\beta \dist_*(\tilde{x},x)
\\ &\le& C^\supReturn |x_u|^\beta (1+\vartheta) | \tilde{x}_c - x_c |
\\ &\le& C^\supReturn |x_u|^\beta (1+\vartheta) K_\subC^{-1} 
         | \MapCombined(\tilde{x})_c - \MapCombined(x)_c |.
\end{array}
\]
The last inequality uses the expansion in $\xc$, 
thus it is valid for $\vartheta \le 1/\sigma$.
We choose $\delta$ small enough 
such that $C^\supReturn |x_u|^\beta K_\subC^{-1} < \sigma/(1+\sigma)$.
Due to the monotone increase of $\vartheta/(1+\vartheta)$ we also have 
$C^\supReturn |x_u|^\beta K_\subC^{-1} < \vartheta/(1+\vartheta)$ 
for all $\vartheta \ge \sigma$.
Thus we obtain the cone invariance
\begin{equation}\label{eqEstimateCone}
\dist_*(\MapCombined(\tilde{x})_\subUS , \MapCombined(x)_\subUS ) 
\;\le\; \vartheta | \MapCombined(\tilde{x})_\subC - \MapCombined(x)_\subC |
\end{equation}
for all $\sigma \le \vartheta \le 1/\sigma$.

The choice $\vartheta=\sigma$ yields (forward) invariance 
of the cone $C_x^\subC$ 
and the choice $\vartheta = 1/\sigma$ yields (backward) invariance 
of the cone $C_{\MapCombined x}^\subUS$.
Note that the cone invariances are in fact strict as claimed in the lemma.
The above estimates are strict inequalities for $x\ne\tilde{x}$.

Now consider the cone in transverse direction, 
that is $\MapCombined(\tilde{x}) \in C_{\MapCombined x}^\subUS$,
which amounts to 
$|\MapCombined(\tilde{x})_\subC - \MapCombined(x)_\subC| \le 
\sigma \dist_*( \MapCombined(\tilde{x})_\subUS , \MapCombined(x)_\subUS )$. 
We have already established invariance. 
Thus $|\tilde\xc-\xc| \le \sigma \|\tilde\xus-\xus\|$ 
and estimate (\ref{eqEstimateReturnMap}) yields
\[
\begin{array}{rcl}
\dist_*( \MapCombined(\tilde{x})_\subUS , \MapCombined(x)_\subUS )
   &\le& C^\supReturn |x_u|^\beta \dist_*(\tilde{x},x)
\\ &\le& C^\supReturn |x_u|^\beta (1+\sigma) \dist_*(\tilde\xus,\xus).
\end{array}
\] 
This is the claimed contraction, 
$K_\subUS^{-1} = C^\supReturn \delta^\beta (1+\sigma) < 1$, 
for $\delta$ small enough.
\end{proof}

\begin{thm}\label{thStableSet}
Assume conditions \ref{condLocalPassage} on the local passages and 
conditions \ref{condGlobalExcursion} on the global excursions.

The (local) stable set of the origin under the sequence of 
return maps $\MapCombined_k$ is given by
\[
\mathcal{W}_k^\supLocal = \{ (\xu^k, \xs{\ell}^k, \xc^k) \,|\, \xc^k = \xc^k(\xu^k, \xs{\ell}^k) \}.
\]
The functions $\xc^k$ are Lipschitz continuous with respect to the metric (\ref{eqMetricIn}). 
Furthermore, $\xc^k(0) = 0$ and 
$\MapCombined(\mathcal{W}_k^\supLocal) \subset \mathcal{W}_{k+1}^\supLocal$.
\end{thm}

\begin{proof}
The idea of the proof is to define a graph transformation on the space of 
sequences of Lipschitz-continuous graphs 
$\{\xus\mapsto\xc=\zeta_k(\xus))\;|\;k\in\setN\}$ 
by the inverse return maps $\MapCombined_k^{-1}$.
The uniform cone invariance provided by the previous lemma will ensure 
that the Lipschitz property of the graphs is preserved.
Due to the expansion/contraction conditions of the previous lemma, 
the graph transformation turns out to be a contraction
on the space of sequences of Lipschitz-continuous graphs. 
The fixed point of this contraction then yields the claim.

To make this idea precise, consider the Banach space of Lipschitz-continuous 
functions 
\[
\begin{array}{rcrll}
 X &=& \{ &
             \zeta\,:\;(0,\delta]^{1+N}\cup{0}\to[-\delta,\delta],\;
           \xus=(\xu,\xs{}) \mapsto \xc=\zeta(\xus)
\\ &&& \mbox{such that } \mathrm{Lip}(\zeta)\le\sigma 
       \;\mbox{ and }\; \zeta(0)=0       & \}
\end{array}
\]
with sup-norm. The parameters $\delta,\sigma < 1$ correspond 
to those of lemma \ref{thLipschitzReturnMap}.
Lipschitz continuity is considered with respect to the metric $\dist_*$ 
given by (\ref{eqMetricIn}).
Consider also the space of sequences
\[
X^\setN \;=\; \{ \; (\zeta_k)_{k\in\setN} \; | \; \zeta_k \in X \;\}
\]
with sup-norm.

Define maps $G_k:X\to X$ as 
$\mathrm{graph}(G_k\zeta_{k+1}) := \MapCombined_k^{-1}\mathrm{graph}(\zeta_{k+1})$,
i.e. as the transformations of the graphs of the functions in $X$.
More precisely
\[
\begin{array}{rcll}
G_k \zeta\left(\left(\MapCombined_k^{-1}(\xus,\zeta(\xus))\right)_\subUS\right) 
  &:=& \left(\MapCombined_k^{-1}(\xus,\zeta(\xus))\right)_c,
  & \mbox{for } \xus \ne 0,
\\
G\zeta(0) &:=& 0.
\end{array}
\]
The first equation implicitly assumes that $(\xus,\zeta(\xus))$ 
has a pre-image under $\MapCombined$ and that it lies in the domain.
The second equation just gives the pre-image of the origin under $\MapCombined$.
Note the restriction to non-negative $\xu,\xs{1},\xs{N}$ consistent with 
the invariant boundaries (loc-v), (glob-i).

We will prove the following claims, uniformly in the index $k$, 
(which is dropped from now on to simplify notation)
\begin{enumerate}
\item[(i)] 
domain of definition:
for all $\zeta\in X$ and $\xus \in (0,\delta]^{N+1}$ 
there exists $\tilde\xus \in (0,\delta]^{N+1}$, 
such that
$\left(\MapCombined^{-1}(\tilde\xus,\zeta(\tilde\xus))\right)_\subUS = \xus$.
\item[(ii)]
well-definedness:
for all $\zeta\in X$ and $\xus, \tilde\xus \in (0,\delta]^{N+1}$ the following holds.
If $\left(\MapCombined^{-1}(\xus,\zeta(\xus))\right)_\subUS = 
\left(\MapCombined^{-1}(\tilde\xus,\zeta(\tilde\xus))\right)_\subUS \in (0,\delta]^{N+1}$
then already $\xus = \tilde\xus$.
\end{enumerate}
Conditions (i) and (ii) yield a well defined function $G\zeta$ with 
$(G\zeta)(0)=0$ for every $\zeta\in X$. 
\begin{enumerate}
\item[(iii)] 
Lipschitz property:
for all $\zeta\in X$ the function $G\zeta$ is again Lipschitz continuous 
with Lipschitz constant $\mathrm{Lip}(G\zeta)\le \sigma$.
Note that the Lipschitz property is again considered 
with respect to the metric $\dist_*$. 
\item[(iv)] 
contraction: 
The exists a constant $0<\kappa<1$ such that
for all $\zeta,\tilde{\zeta} \in X$ the estimate 
$\|G\tilde{\zeta} - G\zeta \|_\mathrm{sup} \le 
\kappa \| \tilde{\zeta} - \zeta \|_\mathrm{sup}$ 
holds.
\end{enumerate}
Conditions (i)--(iii) prove that the graph transformation $G$ 
indeed maps Lipschitz continuous functions in $X$ to 
Lipschitz continuous functions in $X$, with respect to the metric $\dist_*$.
Condition (iv) provides a contraction.
Uniformity of bounds yield a contraction on the space $X^\setN$ of sequences.
If all four conditions hold, then by contraction-mapping theorem 
there is a unique fixed point, 
i.e.~a sequence of Lipschitz continuous function $\zeta_k^*\in X$ with 
$G_k\zeta_{k+1}^* = \zeta_k^*$.

Its graphs form a forward invariant set under $\MapCombined$ composed of local 
manifold. It is also the stable set of the origin due 
to the cone conditions of lemma \ref{thLipschitzReturnMap}. 
This yields the claim of the theorem.
Therefore it remains to prove (i)--(iv):

(i) Let $\zeta\in X$ and $\xus \in (0,\delta]^{N+1}$ be given. 
The straight line $\{\xus\} \times [-\delta,\delta]$ is contained in the cone
$C_{(\xus,0)}^\subC$. 
We use lemma \ref{thLipschitzReturnMap}: 
$\MapCombined(\xus,0) \in \tilde\SectionIn$ by invariance and contraction 
of the cone $C_0^\mathrm{u,s}$.
Thus, by invariance and Expansion of $C_{(\xus,0)}^\subC$, 
the image of the straight line $\{\xus\} \times [-\delta,\delta]$ under 
$\MapCombined$ contains a curve in $C_{\MapCombined(\xus,0)}^\subC$ 
connecting the extremal planes $\{\xc=\pm\delta\}$. 
By the intermediate value theorem this curve must intersect the graph of 
$\zeta$.

(ii) Let $\zeta\in X$ and $\xus,\tilde\xus \in (0,\delta]^{N+1}$ be given 
with $\left(\MapCombined^{-1}(\xus,\zeta(\xus))\right)_\subUS = 
\left(\MapCombined^{-1}(\tilde\xus,\zeta(\tilde\xus))\right)_\subUS 
\in (0,\delta]^{N+1}$. 

Then $\MapCombined^{-1}(\tilde\xus,\zeta(\tilde\xus)) \in 
C_{\MapCombined^{-1}(\xus,\zeta(\xus))}^\subC$,
and by cone invariance 
$(\tilde\xus,\zeta(\tilde\xus)) \in C_{(\xus,\zeta(\xus))}^\subC$.
The Lipschitz-bound on $\zeta\in X$ on the other hand implies 
$(\tilde\xus,\zeta(\tilde\xus)) \in C_{(\xus,\zeta(\xus))}^\subUS$,
thus $(\tilde\xus,\zeta(\tilde\xus)) = (\xus,\zeta(\xus))$.

(iii) Again, the Lipschitz-bound on $\zeta\in X$ translates to 
$(\tilde\xus, \zeta(\tilde\xus)) \in C_{(\xus, \zeta(\xus))}^\subUS$ 
for all $x,\tilde{x}$.
Cone invariance and lemma \ref{thLipschitzReturnMap}, 
immediately yield the Lipschitz bound on $G\zeta$.

(iv) The origin is fixed by construction, thus we only have to estimate 
the distance of the nonsingular part.
Let $\zeta,\tilde{\zeta}\in X$ and $\xus,\tilde\xus \in (0,\delta]^{N+1}$ 
be given with $(\MapCombined^{-1}(\xus,\zeta(\xus)))_\subUS = 
(\MapCombined^{-1}(\tilde\xus,\tilde{\zeta}(\tilde\xus)))_\subUS 
\in (0,\delta]^{N+1}$.

Again, this implies $\MapCombined^{-1}(\tilde\xus,\tilde{\zeta}(\tilde\xus)) 
\in C_{\MapCombined^{-1}(\xus,\zeta(\xus))}^\subC$,
and by cone invariance we have 
$(\tilde\xus,\tilde{\zeta}(\tilde\xus)) \in C_{(\xus,\zeta(\xus))}^\subC$.
Thus we can estimate
\[
\begin{array}{rcl}
|\tilde{\zeta}(\tilde\xus) - \zeta(\xus)| 
   &\le& \|\tilde{\zeta} - \zeta\|_\mathrm{sup} 
         + \sigma \dist_*(\tilde\xus, \xus)
\\ &\le& \|\tilde{\zeta} - \zeta\|_\mathrm{sup} 
         + \sigma^2 |\tilde{\zeta}(\tilde\xus) - \zeta(\xus)|
\end{array}
\]
The first inequality uses the Lipschitz bound on $\zeta\in X$ 
whereas the second one uses the aforementioned cone $C_{(\xus,\zeta(\xus))}^\subC$.
We obtain
\[
|\tilde{\zeta}(\tilde\xus) - \zeta(\xus)|  \;\le\;  
\frac{1}{1-\sigma^2} \|\tilde{\zeta} - \zeta\|_\mathrm{sup}.
\]
On the other hand, the expansion of 
$C_{\MapCombined^{-1}(\xus,\zeta(\xus))}^\subC$ under $\MapCombined$ yields
\[
\begin{array}{rcl}
&& \hspace{-5em} 
\left| (G\tilde{\zeta} - G\zeta) 
   \left(\left(\MapCombined^{-1}(\xus,\zeta(\xus))\right)_\subUS\right) \right|
\\ &=& 
\left| \left(\MapCombined^{-1}(\xus,\zeta(\xus))\right)_\subC 
   - \left(\MapCombined^{-1}(\tilde\xus,\tilde{\zeta}(\tilde\xus))\right)_\subC \right|
\\ &\le&
\frac{1}{K_\subC} \left| \zeta(\xus) - \tilde{\zeta}(\tilde\xus) \right|
\\ &\le&
\frac{1}{K_\subC (1-\sigma^2)} \|\tilde{\zeta} - \zeta\|_\mathrm{sup}.
\end{array}
\]
Lemma \ref{thLipschitzReturnMap} provides constants $K_\subC$, $\sigma$ 
with $K_\subC(1-\sigma^2) > 1$. 
Therefore the last estimates yield the claimed contraction, 
$\kappa = 1/(K_\subC(1-\sigma^2))$, and this finishes the proof.
\end{proof}

With theorem \ref{thStableSet} we have finally proved the main theorems 
\ref{thPeriod3Manifold}, \ref{thAbstractSequenceManifold} as 
formulated in section \ref{secMainResult}.


\section{Discussion and outlook}
\label{secDiscussion}

Unfortunately the metric used to obtain the contraction in the proof of
the main theorem of this paper is singular on the invariant subspaces. We have 
no result on the way in which the manifolds constructed approach these 
boundaries. The only exception is the heteroclinic cycle itself. The cap of 
heteroclinic orbits corresponds to the line $\{0,\xc\}$ and the new metric is 
regular there. In fact, the manifold is nicely attached to the given 
heteroclinic chain. Moreover, in the proof of theorem \ref{thStableSet} we 
could restrict to very small neighbourhoods of the primary heteroclinic chain, 
i.e. $\delta\to0$. Then we can choose arbitrarily small Lipschitz bounds on the 
functions considered, i.e. $\sigma\to0$. Thus the manifolds constructed are 
tangent to the fibre $\{\xc \;=\mbox{constant}\}$ at the heteroclinic chain. 

For completeness the following subtlety should be mentioned. The set of
points in the domain of definition of the dynamical system corresponding to 
Bianchi type IX vacuum solutions or Bianchi type VI${}_0$ solutions with 
magnetic field is an open subset bounded by invariant manifolds and it
lies on only one side of these manifolds. The fact that the global excursion
map has its image on the correct side of these manifolds is not mentioned
in the analytical treatment above. Nevertheless it follows immediately from
the nature of the underlying geometrical problem.  

Up to this point vacuum models of type IX were replaced by Einstein-Maxwell
models of type VI${}_0$ and one non-vanishing magnetic field component. Now
some generalizations will be mentioned. 
In \cite{LiebscherHaerterichWebsterGeorgi2011-BianchiA} some results
were obtained for type IX solutions with perfect fluids having a linear
equation of state $p=(\gamma-1)\rho$. Restrictions had to be imposed on the
value of $\gamma$. The techniques developed in this paper allow these 
results to be generalized to cases where these restrictions are relaxed.
The four-dimensional dynamical system is replaced by a five-dimensional
one and at each Kasner point there is an additional eigenvalue $3(2-\gamma)$.
This situation can be treated for all $\gamma<2$. In particular the method
applies for all values of $\gamma$ in the physical range $[1,2]$ except
for the case $\gamma=2$ where the dynamics is known to be very different.
Bianchi type VI${}_0$ solutions with a perfect fluid and a magnetic field can 
be treated in a very similar way. The additional eigenvalue arising from the 
fluid is the same as in the case without magnetic field 
\cite{LeBlancKerrWainwright1995-MagneticBianchiVI}.

It is possible to formulate the Bianchi type II models with a magnetic
field as a five-dimensional dynamical system 
\cite{LeBlanc1997-MagneticBianchiI}. 
In this approach the magnetic field has only one non-zero component in 
the frame used but the metric has a non-zero off-diagonal component in that 
frame. The eigenvalues of the linearisation about a Kasner solution are given 
by
\begin{equation}
3p_1, \quad 6p_2, \quad 3(p_3-p_1).
\end{equation}
An important qualitative difference to the model of type VI${}_0$ is 
that for some regions of the Kasner circle the stable manifold of 
the Kasner solution is two-dimensional. Thus in general the results of 
this paper do not apply to heteroclinic chains for the Bianchi type II model 
with magnetic field. In a similar way it is possible to 
formulate the type I models with a magnetic field as a five-dimensional 
dynamical system with only one component of the magnetic field being non-zero
\cite{LeBlanc1997-MagneticBianchiI}. 
In this case all the off-diagonal metric components are non-zero in general.
Again it happens that the stable manifold can be two-dimensional. 
Note that it has been shown in 
\cite{LeBlancKerrWainwright1995-MagneticBianchiVI} that it is not 
possible to have solutions of the Einstein-Maxwell equations of Bianchi type 
VIII or IX with a non-vanishing pure magnetic field.

Up to now there is no generalization of the results of 
\cite{LiebscherHaerterichWebsterGeorgi2011-BianchiA} to oscillatory models
of Bianchi class B. In fact it would be very interesting to have such
results for Bianchi type VI${}_{-\frac19}$ where oscillatory solutions are
expected to exist. One obstacle is the existence of stable manifolds of
dimension greater than one as in the examples with magnetic field above.
Another is that invariant manifolds of the type which played such an 
important role in the proofs of this paper do not appear to exist for
models of Bianchi class B. 

In the case of Bianchi type IX vacuum models it has been proved that the 
$\alpha$-limit set of each solution belongs to the union of points of type I
and type II \cite{Ringstroem2001-BianchiIX}. Interestingly it is not known if 
the corresponding statement holds for the superficially similar type VIII. 
This contrasts with the fact that the results for type IX in
\cite{LiebscherHaerterichWebsterGeorgi2011-BianchiA} extend
almost without change to type VIII. It is easy to formulate an analogue 
of the result of \cite{Ringstroem2001-BianchiIX} for solutions of type 
VI${}_0$ with magnetic field and it would be interesting to investigate 
whether it holds, especially since this might throw some new light on the 
unsolved Bianchi VIII problem.

To sum up, it is clear that the above complex of problems represents a 
promising opportunity to learn about the related questions of the BKL 
conjecture, the dynamics of Bianchi models near the initial singularity
and the stability of heteroclinic cycles in more general dynamical systems.

\smallskip

\noindent\textbf{Acknowledgement:} 
The research of SL and ADR was partially supported by the Collaborative Research
Centre 647 Space--Time--Matter of the German Research Foundation (DFG). SBT
thanks the Albert Einstein Institute for hospitality during a one month visit at
the start of this work.


\bibliographystyle{alpha}
\bibliography{LieRenTch-MaxwellBianchi}

\begin{thebibliography}{LHWG11}

\bibitem[B{\'e}g10]{Beguin2010-BianchiAsymptotics}
F.~B{\'e}guin.
\newblock Aperiodic oscillatory asymptotic behavior for some {B}ianchi
  spacetimes.
\newblock {\em Classical Quantum Gravity}, 27:185005, 2010.

\bibitem[BKL70]{BelinskiiKhalatnikovLifshitz1970-OscillatoryApproach}
V.A. Belinskii, I.M. Khalatnikov, and E.M. Lifshitz.
\newblock Oscillatory approach to a singular point in the relativistic
  cosmology.
\newblock {\em Adv.~Phys.}, 19:525--573, 1970.

\bibitem[BKL82]{BelinskiiKhalatnikovLifshitz1982-GeneralSolution}
V.A. Belinskii, I.M. Khalatnikov, and E.M. Lifshitz.
\newblock A general solution of the {E}instein equations with a time
  singularity.
\newblock {\em Adv.~Phys.}, 31:639--667, 1982.

\bibitem[HU09]{HeinzleUggla2009-Mixmaster}
J.M. Heinzle and C.~Uggla.
\newblock Mixmaster: Fact and belief.
\newblock {\em Classical Quantum Gravity}, 26(7):075016, 2009.

\bibitem[LeB97]{LeBlanc1997-MagneticBianchiI}
V.G. LeBlanc.
\newblock Asymptotic states of magnetic {B}ianchi {I} cosmologies.
\newblock {\em Classical Quantum Gravity}, 14:2281--2301, 1997.

\bibitem[LHWG11]{LiebscherHaerterichWebsterGeorgi2011-BianchiA}
S.~Liebscher, J.~H{\"a}rterich, K.~Webster, and M.~Georgi.
\newblock Ancient dynamics in {B}ianchi models: Approach to periodic cycles.
\newblock {\em Communications in Mathematial Physics}, 305:59--83, 2011.

\bibitem[LKW95]{LeBlancKerrWainwright1995-MagneticBianchiVI}
V.G. LeBlanc, D.~Kerr, and J.~Wainwright.
\newblock Asymptotic states of magnetic {B}ianchi {VI${}_0$} cosmologies.
\newblock {\em Classical Quantum Gravity}, 12:513--541, 1995.

\bibitem[Ren08]{Rendall2008-Book}
A.D. Rendall.
\newblock {\em Partial differential equations in general relativity}.
\newblock Oxford University Press, Oxford, 2008.

\bibitem[Rin00]{Ringstroem2000-CurvatureBlowup}
H.~Ringstr{\"o}m.
\newblock Curvature blow up in {B}ianchi {VIII} and {IX} vacuum solutions.
\newblock {\em Classical Quantum Gravity}, 17:713--731, 2000.

\bibitem[Rin01]{Ringstroem2001-BianchiIX}
H.~Ringstr{\"o}m.
\newblock The {B}ianchi {IX} attractor.
\newblock {\em Ann.~Henri Poincar{\'e}}, 2(3):405--500, 2001.

\bibitem[RT10]{ReitererTrubowitz2010-BKL}
M.~Reiterer and E.~Trubowitz.
\newblock The {BKL} conjectures for spatially homogeneous spacetimes.
\newblock arXiv:1005.4908, 2010.

\bibitem[SSTC98]{ShilnikovTuraevChua1998-MethodsNonlinearDynamicsI}
L.P. Shilnikov, A.L. Shilnikov, D.V. Turaev, and L.O. Chua.
\newblock {\em Methods of Qualitative Theory in Nonlinear Dynamics {I}},
  volume~4 of {\em Series on Nonlinear Science, Series A}.
\newblock World Scientific, 1998.

\bibitem[WE97]{WainwrightEllis1997-Cosmology}
J.~Wainwright and G.F.R. Ellis, editors.
\newblock {\em Dynamical systems in cosmology}.
\newblock Cambridge University Press, Cambridge, 1997.

\bibitem[Wea00]{Weaver2000-MagneticBianchi}
M.~Weaver.
\newblock Dynamics of magnetic {B}ianchi type {VI${}_0$} cosmologies.
\newblock {\em Classical Quantum Gravity}, 17:421--434, 2000.

\bibitem[WH89]{WainwrightHsu1989-BianchiA}
J.~Wainwright and L.~Hsu.
\newblock A dynamical systems approach to {B}ianchi cosmologies: orthogonal
  models of a class {A}.
\newblock {\em Classical Quantum Gravity}, 6(10):1409--1431, 1989.

\end{thebibliography}

\end{document}